\documentclass[journal,twoside,web]{ieeecolor}
\usepackage{generic}
\usepackage{cite}
\usepackage{amsmath,amssymb,amsfonts}
\usepackage{algorithmic}
\usepackage{graphicx}
\usepackage{textcomp}
\let\proof\relax

\usepackage{amsthm}
\usepackage{algorithmic}
\usepackage{algorithm}
\usepackage{array}
\usepackage[caption=false,font=normalsize,labelfont=sf,textfont=sf]{subfig}
\usepackage{stfloats}
\usepackage{url}
\usepackage{verbatim}
\usepackage[justification=centering]{caption}
\usepackage{multirow}
\usepackage{makecell}
\usepackage{color}
\def\BibTeX{{\rm B\kern-.05em{\sc i\kern-.025em b}\kern-.08em
    T\kern-.1667em\lower.7ex\hbox{E}\kern-.125emX}}
\begin{document}
\newtheorem{remark}{Remark}
\newtheorem{theorem}{Theorem}
\newtheorem{corollary}{Corollary}
\newtheorem{lemma}{Lemma}
\newtheorem{definition}{Definition}
\newtheorem{assumption}{Assumption}
\renewcommand{\proof}{\textbf{Proof}.}

\newcommand{\KP}[1]{{\color{blue} #1}}
\newcommand{\TY}[1]{{\color{red} #1}}
\definecolor{gray}{RGB}{128,128,128}

\title{\LARGE \bf One-Point Sampling for Distributed Bandit Convex Optimization with Time-Varying Constraints}

\author{Kunpeng Zhang$^{1}$, Lei Xu$^{2}$,~\IEEEmembership{Member,~IEEE}, Xinlei Yi$^{3}$, Guanghui Wen$^{4}$,~\IEEEmembership{Senior Member,~IEEE}, \\Lihua Xie$^{5}$,~\IEEEmembership{Fellow,~IEEE}, Tianyou Chai$^{1}$,~\IEEEmembership{Life Fellow,~IEEE}, and Tao Yang$^{1}$,~\IEEEmembership{Senior Member,~IEEE}
\thanks{This work was supported by the National Key Research and Development Program of China under Grant 2022YFB3305904 and the National Natural Science Foundation of China under Grant 62133003. (Corresponding author: Tao Yang.)}% <-this % stops a space
\thanks{$^{1}$K. Zhang, T. Chai and T. Yang are with the State Key Laboratory of Synthetical Automation for Process Industries, Northeastern University, Shenyang 110819, China {\tt\small 2110343@stu.neu.edu.cn; \{tychai; yangtao\}@mail.neu.edu.cn}}
\thanks{$^{2}$L. Xu is with the Department of Mechanical Engineering, University of Victoria, Victoria, BC V8W 2Y2, Canada {\tt\small leix@uvic.ca}}
\thanks{$^{3}$X. Yi is with Shanghai Institute of Intelligent Science and Technology, National Key Laboratory of Autonomous Intelligent Unmanned Systems, and Frontiers Science Center for Intelligent Autonomous Systems, Ministry of Education, Shanghai 201210, China {\tt\small xinleiyi@tongji.edu.cn}}
\thanks{$^{4}$Guanghui Wen is with the Department of Mathematics, Southeast University, Nanjing 210096, China {\tt\small  ghwen@seu.edu.cn}}
%\thanks{M. Cao is with the Engineering and Technology Institute Groningen, Faculty of Science and Engineering, University of Groningen, AG 9747 Groningen, The Netherlands (e-mail: m.cao@rug.nl).}%
\thanks{$^{5}$L. Xie is with the School of Electrical and Electronic Engineering, Nanyang Technological University, 50 Nanyang Avenue,  Singapore 639798 {\tt\small elhxie@ntu.edu.sg}}
}

\maketitle

\begin{abstract}
This paper considers the distributed bandit convex optimization problem with time-varying constraints. In this problem, the global loss function is the average of all the local convex loss functions, which are unknown beforehand.
Each agent iteratively makes its own decision subject to time-varying inequality constraints which can be violated but are fulfilled in the long run.
For a uniformly jointly strongly connected time-varying directed graph, a distributed bandit online primal--dual projection algorithm with one-point sampling is proposed.
We show that sublinear dynamic network regret and network cumulative constraint violation are achieved if the path--length of the benchmark also increases in a sublinear manner. In addition, an $\mathcal{O}({T^{3/4 + g}})$ static network regret bound and an $\mathcal{O}( {{T^{1 - {g}/2}}} )$ network cumulative constraint violation bound are established, where $T$ is the total number of iterations and $g  \in ( {0,1/4} )$ is a trade-off parameter.
Moreover, a reduced static network regret bound $\mathcal{O}( {{T^{2/3 + 4g /3}}} )$ is established for strongly convex local loss functions.
Finally, a numerical example is presented to validate the theoretical results.
\end{abstract}


\section{Introduction}
\label{sec:introduction}
\IEEEPARstart{O}{ver} the past decades, the bandit convex optimization (BCO) framework has garnered substantial interest owing to its wide applications including real-time routing and advertisement placement \cite{Li2023}.
The framework is a sequential decision making process spanning $T$ iterations in dynamic environments \cite{Hazan2016a}. Specifically, at iteration $t$, a decision maker selects ${x_t}$ from a convex set $\mathbb{X}$ in Euclidean space. Following that, the value of a convex loss function ${l_t} :\mathbb{X} \to \mathbb{R}$ at ${x_t}$ is revealed, where $\mathbb{R}$ denotes the set of all real numbers. Accordingly, the decision maker suffers a loss ${l_t}( {{x_t}} )$.
The objective is to minimize the cumulative loss across $T$ iterations, denoted by $\sum\nolimits_{t = 1}^T {{l_t}( {{x_t}} )}$.
In general, static regret is leveraged to measure an algorithm given rise by the framework, denoted by\par\nobreak\vspace{-10pt}
\begin{small}
\begin{flalign}
\nonumber
\sum\limits_{t = 1}^T {{l_t}( {{x_t}} )}  - \mathop {\min }\limits_{_{x \in \mathbb{X}}} \sum\limits_{t = 1}^T {{l_t}( x )},
\end{flalign}
\end{small}%
which measures the difference of the cumulative loss between the decision sequence induced by the algorithm and the optimal static decision with the benefit of hindsight.

Various sampling-based bandit online projection algorithms are proposed to solve the BCO problem, see, e.g., \cite{Flaxman2005, Dani2008, Abernethy2009, Saha2011, Hazan2014, Duchi2015, Yang2016a, Shamir2017, Agarwal2010, Tatarenko2018, Shames2019}.
For example,
by modifying the gradient-based online projection algorithm in \cite{Zinkevich2003}, the authors of \cite{Flaxman2005} develop a bandit online projection algorithm with one-point sampling, i.e., the decision maker samples the loss function at only one point
at each iteration, and an ${\cal O}({T^{3/4}})$ static regret bound is established for convex loss functions.
For strongly convex functions, the authors of \cite{Saha2011} establish an $\mathcal{O}( {{T^{2/3}}} )$ static regret bound.
By using the two-point sampling strategy, i.e., the decision maker samples the loss function at two points
at each iteration, the authors of \cite{Shamir2017} establish an $\mathcal{O}( {{T^{1/2}}} )$ static regret bound for convex loss functions.
By using the multi-point sampling strategy, the authors of \cite{Agarwal2010} establish an $\mathcal{O}\big( {\log ( T )} \big)$  static regret bound for strongly convex loss functions.

Note that in the aforementioned studies, the convex set~$\mathbb{X}$ is a simple closed set, e.g., a box or ball.
However, inequality constraints commonly exist in various applications, see, e.g., \cite{Mahdavi2012, Yuan2018}.
To cope with this scenario, the authors of \cite{Mahdavi2012} characterize the convex set $\mathbb{X}$ as the intersection of a simple closed set and an static inequality constraint set. The challenge is that the projection operation on the intersection at each iteration imposes severe computation burden.
To reduce the burden, long term constraints are considered, where the inequality constraints are permitted to be violated but are fulfilled in the long run.
To measure cumulative violation of the inequality constraints, constraint violation, is given by\par\nobreak\vspace{-10pt}
\begin{small}
\begin{flalign}
\Big\| {{{\Big[ {\sum\limits_{t = 1}^T {c( {{x_t}} )} } \Big]_ +} }} \Big\|, \label{introduction-CV-eq1}
\end{flalign}
\end{small}%
where $\|  \cdot  \|$ denotes Euclidean norm for vectors, $[  \cdot  ]_ + $ is the projection onto the nonnegative space, $c:{\mathbb{X}} \to {\mathbb{R}^m}$ is the static constraint function, and $m$ is a positive integer.
Based on the idea of long term constraints, the authors of \cite{Chen2019,Cao2019} study the BCO problem with time-varying inequality constraints.
The constraint violation metric \eqref{introduction-CV-eq1} allows the compensation of violated constraints at one iteration by strictly satisfied constraints at other iterations. That is inapplicable to some applications where constraints have no cumulative nature, such as safety-critical applications \cite{Guo2022}. To tackle this challenge, the authors of \cite{Yi2021a} propose a stricter form of constraint violation metric, namely cumulative constraint violation (CCV),\par\nobreak\vspace{-10pt}
\begin{small}
\begin{flalign}
\sum\limits_{t = 1}^T {\| {{{[ {c_t( {{x_t}} )} ]}_ + }} \|}, \label{introduction-CCV-eq2}
\end{flalign}
\end{small}%
where $c_t:{\mathbb{X}} \to {\mathbb{R}^m}$ is the constraint function at iteration~$t$. The constraint violation metric \eqref{introduction-CCV-eq2} implies that all the violations of inequality constraints are taken into consideration.

The aforementioned studies focus on centralized BCO, which encounter performance limitations, e.g., single point of the failure, heavy communication and computation overhead \cite{Nedic2018, Yang2019}.
To deal with the limitations, distributed bandit convex optimization (DBCO) problem is studied, see \cite{Wang2020, Cao2021, Li2021a, Yuan2021a, Yuan2021b, Tu2022, Patel2022, Xiong2023}.
In this problem, the loss function at each iteration is decomposed across a network of agents, and local decisions of the agents are the copies of global decisions.
%For example, the authors of \cite{Li2020} propose a distributed bandit online primal--dual projection algorithm with two-point sampling, where local decisions of all decision makers are a portion of global decisions.
Recently, the authors of \cite{Yuan2022} extend this problem by incorporating long term constraints, where inequality constraints are static. Additionally, they propose a distributed bandit online primal--dual projection algorithm by using the one-point sampling strategy. They establish an $\mathcal{O}\big( {{T^{\max \{ {g, 1 - g/3} \}}}} \big)$ network regret bound and an $\mathcal{O}( {{T^{1 - g/2}}} )$ network CCV bound with $g \in ( {0,1} )$ for convex local loss and constraint functions. Furthermore, for strongly convex local loss functions, an $\mathcal{O}\big( {{T^{2/3}}\log ( T )} \big)$ network regret bound and an $\mathcal{O}\big( {\sqrt {T\log ( T )} } \big)$ network CCV bound are established.
Different from the considered DBCO problem with long term constraints in \cite{Yuan2022}, the authors of \cite{Yi2023} consider more challenging time-varying constraints where inequality constraints are time-varying rather than static. Moreover, they propose a distributed bandit online primal--dual projection algorithm by using the two-point sampling strategy, and establish a reduced $\mathcal{O}\big( {{T^{\max \{ {g, 1 - g} \}}}} \big)$ network regret bound and an $\mathcal{O}( {{T^{1 - g/2}}} )$ network CCV bound for convex local loss functions, and a reduced $\mathcal{O}( {{T^g}} )$ network regret bound and an $\mathcal{O}( {{T^{1 - g/2}}} )$ network CCV bound for strong convex local loss functions, where $g \in ( {0,1} )$.
However, the two-point sampling strategy increases the sampling, computational, and memory requirements. Moreover, algorithms that use the two-point sampling strategy are generally easier to analyze in terms of network regret and CCV bounds compared to those utilizing the one-point strategy, as the former algorithms have more information.
It is worth noting that the one-point sampling strategy is more naturally suited to the BCO framework in~\cite{Hazan2016a}, as each agent selects a single decision at each iteration.

In time-varying inequality constraint setting, a distributed bandit online algorithm with one-point sampling remains missing. That motivates this paper. Therefore, we consider the DBCO problem with time-varying constraints where inequality constraints are time-varying and unknown in advance.
The contributions are distilled into the following. \vspace{-1pt}
\begin{itemize}
\item[$\bullet$]
By integrating one-point sampling strategy with the algorithm with two-point sampling in \cite{Yi2023}, we propose a distributed bandit online primal--dual projection algorithm with one-point sampling.
Compared to the algorithms with one-point sampling in \cite{Flaxman2005, Saha2011, Chen2019, Yuan2022}, the proposed algorithm does not require knowledge of the total number $T$ of iterations.
More importantly, we consider time-varying inequality constraints and one-point bandit feedback for local constraint functions, while \cite{Flaxman2005, Saha2011} do not consider inequality constraints, and the algorithms in \cite{Chen2019, Yuan2022} use full-information feedback for time-varying local constraint functions and static local constraint functions, respectively. Note that the proposed algorithm is more challenging to analyze the performance than the algorithm in \cite{Yuan2022} and the algorithm with two-point sampling in \cite{Yi2023} due to one-point sampling for both local loss and constraint functions, which will be explained in detail in Remark~2. Additionally, the update rules of the proposed algorithm are different from the algorithms in~\cite{Yuan2022, Yi2023}.
\item[$\bullet$]
For convex local loss functions, sublinear dynamic network regret and network CCV are achieved if the path--length of the benchmark, i.e., the accumulated dynamic variation of the optimal decision sequence, also increases in a sublinear manner. Moreover, an $\mathcal{O}({T^{3/4 + g}})$ static network regret bound and an $\mathcal{O}( {{T^{1 - {g}/2}}} )$ network CCV bound are established, where $g  \in ( {0,1/4} )$.
In the absence of inequality constraints, an $\mathcal{O}({T^{3/4}})$ static network regret bound is recovered, which is the same as those established by the centralized algorithms with one-point sampling in \cite{Flaxman2005, Chen2019} and the distributed algorithm with one-point sampling in \cite{Yuan2022}. The bound is larger than that established by the distributed algorithm with two-point sampling in \cite{Yi2023}, which is reasonable since the algorithm in \cite{Yi2023} uses more information.
\item[$\bullet$]
For strongly convex local loss functions, reduced dynamic network regret and network CCV are achieved. Moreover, the static network regret bound is reduced to $\mathcal{O}( {{T^{2/3 + 4g /3}}} )$ with ${g} \in ( {0, 1/4} )$.
In the absence of inequality constraints, an $\mathcal{O}({T^{2/3}})$ static network regret bound is recovered, which is the same as the static regret bound established by the centralized algorithm with one-point sampling in \cite{Saha2011}. Again, the bound is larger than that established by the distributed algorithm with two-point sampling in \cite{Yi2023}, which is reasonable since the algorithm in \cite{Yi2023} uses more information.
\end{itemize}

A detailed comparison to related studies is provided in TABLE~I,
where we only present static regret and (cumulative) constraint violation bounds for the sake of clarity.

The remainder of this paper is presented in the following order. Section~II introduces the problem formulation. Section~III proposes the algorithm and evaluates its performance. Section~IV provides a numerical example. Finally, Section~V concludes the paper, and all proofs can be found in Appendix.

\begin{table*}[t]
\centering
\caption{Comparison of this paper to related studies.}
\label{table1}
\begin{tabular*}{\hsize}{@{}@{\extracolsep{\fill}}c|c|c|c|c|c|c|c|c@{}}
\Xcline{1-7}{1pt}
\multicolumn{2}{c|}{\multirow{2}{*}{Reference}} &\multirow{2}{*}{Problem type}& \multirow{2}{*}{Constraint type} & \multirow{2}{*}{Information feedback} & \multicolumn{2}{c}{\multirow{2}{*}{Regret and (cumulative) constraint violation bounds}}\\
\multicolumn{2}{c|}{} & & &  & \multicolumn{2}{c}{}  \\
\cline{1-7}
\multicolumn{2}{c|}{\multirow{3}{*}{\cite{Flaxman2005}}} &{\multirow{3}{*}{\makecell{Centralized \\ and convex}}}	&{\multirow{3}{*}{${c_t}( x ) \equiv {\mathbf{0}_m}$}}	&{\multirow{3}{*}{\makecell{One-point samping \\ for ${l_t}$}}}	&\multicolumn{2}{c}{{\multirow{3}{*}{$\mathcal{O}( {{T^{3/4}}} )$}}} \\
\multicolumn{2}{c|}{} & & &  & \multicolumn{2}{c}{}  \\
\multicolumn{2}{c|}{} & & &  & \multicolumn{2}{c}{}  \\
\cline{1-7}
\multicolumn{2}{c|}{\multirow{3}{*}{\cite{Saha2011}}} &{\multirow{3}{*}{\makecell{Centralized \\ and strongly convex}}}	&{\multirow{3}{*}{${c_t}( x ) \equiv {\mathbf{0}_m}$}}	&{\multirow{3}{*}{\makecell{One-point samping \\ for ${l_t}$}}}	&\multicolumn{2}{c}{{\multirow{3}{*}{$\mathcal{O}( {{T^{2/3}}} )$}}} \\
\multicolumn{2}{c|}{} & & &  & \multicolumn{2}{c}{}  \\
\multicolumn{2}{c|}{} & & &  & \multicolumn{2}{c}{}  \\
\cline{1-7}
\multicolumn{2}{c|}{\multirow{3}{*}{\cite{Chen2019}}} &{\multirow{3}{*}{\makecell{Centralized \\ and convex}}}	&{\multirow{3}{*}{\makecell{${c_t}( x ) \le {\mathbf{0}_m}$ \\ and Slater's condition}}}	&{\multirow{3}{*}{\makecell{One-point samping \\ for ${l_t}$ and $\partial {c_t}$}}}	&\multicolumn{2}{c}{{\multirow{3}{*}{\makecell{$\mathcal{O}({T^{3/4}})$  and $\mathcal{O}( {{T^{3/4}}} )$}}}} \\
\multicolumn{2}{c|}{} & & &  & \multicolumn{2}{c}{}  \\
\multicolumn{2}{c|}{} & & &  & \multicolumn{2}{c}{}  \\
\cline{1-7}
\multicolumn{2}{c|}{\multirow{3}{*}{\cite{Cao2019}}} &{\multirow{3}{*}{\makecell{Centralized \\ and convex}}}	&{\multirow{3}{*}{\makecell{${c_t}( x ) \le {\mathbf{0}_m}$ }}}	&{\multirow{3}{*}{\makecell{Two-point samping \\ for ${l_t}$ and $ {c_t}$}}}	&\multicolumn{2}{c}{{\multirow{3}{*}{\makecell{$\mathcal{O}({\sqrt T })$  and $\mathcal{O}( {{T^{3/4}}} )$}}}} \\
\multicolumn{2}{c|}{} & & &  & \multicolumn{2}{c}{}  \\
\multicolumn{2}{c|}{} & & &  & \multicolumn{2}{c}{}  \\
\cline{1-7}
\multicolumn{2}{c|}{\multirow{6}{*}{\cite{Yuan2022}}} &{\multirow{3}{*}{\makecell{Distributed \\ and convex}}}	&{\multirow{6}{*}{\makecell{$c( x ) \le {\mathbf{0}_m}$}}}	&{\multirow{6}{*}{\makecell{One-point samping \\ for ${l_{i,t}}$ and $\partial c$}}} 	&\multicolumn{2}{c}{{\multirow{3}{*}{\makecell{$\mathcal{O}( {{T^{\max \{ {1 - g/3,g} \}}}})$  and $\mathcal{O}( {{T^{1 - g/2}}} )$, \\ where $g \in ( {0,1} )$}}}} \\
\multicolumn{2}{c|}{} & & &  & \multicolumn{2}{c}{}  \\
\multicolumn{2}{c|}{} & & &  & \multicolumn{2}{c}{}  \\
\cline{3-3}
\cline{6-7}
\multicolumn{2}{c|}{} &{\multirow{3}{*}{\makecell{Distributed \\ and strongly convex}}}	& & &\multicolumn{2}{c}{{\multirow{3}{*}{\makecell{$\mathcal{O}\big( {{T^{2/3}}\log ( T )} \big)$  and $\mathcal{O}\big( {\sqrt {T\log ( T )} } \big)$}}}} \\
\multicolumn{2}{c|}{} & & &  & \multicolumn{2}{c}{}  \\
\multicolumn{2}{c|}{} & & &  & \multicolumn{2}{c}{}  \\
\cline{1-7}
\cline{1-7}
\multicolumn{2}{c|}{\multirow{6}{*}{\cite{Yi2023}}} &{\multirow{3}{*}{\makecell{Distributed \\ and convex}}}	&{\multirow{6}{*}{\makecell{${c_{i,t}}( x ) \le {\mathbf{0}_m}$}}}	&{\multirow{6}{*}{\makecell{Two-point samping \\ for ${l_{i,t}}$ and ${c_{i,t}}$}}} 	&\multicolumn{2}{c}{{\multirow{3}{*}{\makecell{$\mathcal{O}\big( {{T^{\max \{ {g, 1 - g} \}}}} \big)$  and $\mathcal{O}( {{T^{1 - g/2}}} )$, \\ where $g \in ( {0,1} )$}}}} \\
\multicolumn{2}{c|}{} & & &  & \multicolumn{2}{c}{}  \\
\multicolumn{2}{c|}{} & & &  & \multicolumn{2}{c}{}  \\
\cline{3-3}
\cline{6-7}
\multicolumn{2}{c|}{} &{\multirow{3}{*}{\makecell{Distributed \\ and strongly convex}}}	& & &\multicolumn{2}{c}{{\multirow{3}{*}{\makecell{$\mathcal{O}( {{T^g}} )$  and $\mathcal{O}( {{T^{1 - g/2}}} )$, \\ where $g \in ( {0,1} )$}}}} \\
\multicolumn{2}{c|}{} & & &  & \multicolumn{2}{c}{}  \\
\multicolumn{2}{c|}{} & & &  & \multicolumn{2}{c}{}  \\
\cline{1-7}
\multicolumn{2}{c|}{\multirow{8}{*}{This paper}} &{\multirow{4}{*}{\makecell{Distributed \\ and convex}}}	&{\multirow{8}{*}{${c_{i,t}}( x ) \le {\mathbf{0}_m}$}}	&{\multirow{8}{*}{\makecell{One-point samping \\ for ${l_{i,t}}$ and ${c_{i,t}}$}}}		&\multicolumn{2}{c}{{\multirow{4}{*}{\makecell{${\cal O}( {{T^{3/4 + g}}} )$ and $\mathcal{O}( {{T^{1 - {g}/2}}} )$, \\ where ${g} \in ( {0, 1/4} )$}}}} \\
\multicolumn{2}{c|}{} & & &  & \multicolumn{2}{c}{}  \\
\multicolumn{2}{c|}{} & & &  & \multicolumn{2}{c}{}  \\
\multicolumn{2}{c|}{} & & &  & \multicolumn{2}{c}{}  \\
\cline{3-3}
\cline{6-7}
\multicolumn{2}{c|}{} &{\multirow{4}{*}{\makecell{Distributed \\ and strongly convex}}}	& & &\multicolumn{2}{c}{{\multirow{4}{*}{\makecell{${\cal O}( {{T^{2/3 + 4g/3}}} )$ and $\mathcal{O}( {{T^{1 - {g}/2}}} )$, \\ where ${g} \in ( {0, 1/4} )$}}}} \\
\multicolumn{2}{c|}{} & & &  & \multicolumn{2}{c}{}  \\
\multicolumn{2}{c|}{} & & &  & \multicolumn{2}{c}{}  \\
\multicolumn{2}{c|}{} & & &  & \multicolumn{2}{c}{}  \;\;\;\;\;\;\;\;\;\;\;\;\;\;\;\;\;\;\;\;\;\;\;\;\;\;\;\;\;\;\;\;\;\;\;\;\;\;\;\;\;\;\;\;\;\;\;\;\;\;\;\;\;\;\;\;\;\;\;
\;\;\;\;\;\;\;\;\;\;\;\;\;\;\;\;\;\;\;\;\;\;\;\;\;\;\;\;\;\;\;\;\;\;\;\\
\Xcline{1-7}{1pt}
\end{tabular*}
\label{table_MAP}
\end{table*}

\textbf{Notations:} ${\mathbb{R}^p}$ and $\mathbb{R}_ + ^p$ are $p$-dimensional vector set and its nonnegative orthant, respectively. ${\mathbb{N}_ + }$ is the set of all positive integers. ${\|  \cdot  \|_1}$ denotes the $1$-norm for vectors. For a vector $x$, ${x^T}$ is its transpose. $[ m ]$ is the set $\{ {1, \cdot  \cdot  \cdot ,m} \}$ with $m$ being a positive integer. $\mathrm{col}( {q_1}, \cdot  \cdot  \cdot ,{q_n} )$ denotes the concatenated column vector of ${q_i} \in {\mathbb{R}^{{m_i}}}$ for $i \in [ n ]$. ${\mathbb{S}^p}$ denotes the unit sphere. ${\mathbb{B}^p}$ denotes the unit ball. For the vectors $P$ and $Q$, $\langle {P,Q} \rangle $ denotes their inner product. $\otimes$~denotes Kronecker product. ${\mathbf{0}_m} \in {\mathbb{R}^m}$ is the vector whose components are all $0$. ${\mathcal{P}_{\mathbb{K}}}( P ) = \arg {\min _{Q \in {\mathbb{K}}}}{\| {P - Q} \|^2}$ with $\mathbb{K} \subseteq {\mathbb{R}^p}$ and $P \in {\mathbb{R}^p}$. $[  x  ]_+$ denotes ${\mathcal{P}_{\mathbb{R}_ + ^p}}( x )$.
For $l(x):{\mathbb{R}^p} \to \mathbb{R}$, $\partial l( x ) \in {\mathbb{R}^p}$ and $\partial {[ {l( x )} ]_ + } \in {\mathbb{R}^p}$ are the (sub)gradients of $l(x)$ and ${[ l(x) ]_ + }$, respectively.
For $l(x) = {[ {{l_1(x)}, \cdot  \cdot  \cdot ,{l_n(x)}} ]^T}:{\mathbb{R}^p} \to {\mathbb{R}^n}$, $\partial l( x ) = {[ {\partial {l_1}( x ), \cdot  \cdot  \cdot ,\partial {l_n}( x )} ]^T} \in {\mathbb{R}^{p \times n}}$ and $\partial {[ {l( x )} ]_ + } = {\big[ {\partial {{[ {{l_1}( x )} ]}_ + }, \cdot  \cdot  \cdot ,\partial {{[ {{l_n}( x )} ]}_ + }} \big]^T} \in {\mathbb{R}^{p \times n}}$ are the (sub)gradients of $l(x)$ and ${[ l ]_ + }(x)$, respectively.

\section{Problem Formulation}
Consider the DBCO problem with time-varying constraints. In this problem, a network of $n$ agents is modeled by a time-varying directed graph ${\mathcal{G}_t} = ( {\mathcal{V},{\mathcal{E}_t}} )$ with the agent set $\mathcal{V} = [ n ]$ and the edge set ${\mathcal{E}_t} \subseteq \mathcal{V} \times \mathcal{V}$ at iteration~$t$. $( {j,i} ) \in {\mathcal{E}_t}$ indicates that agent $i$ can receive information from agent $j$ at iteration~$t$. In addition, at iteration $t$, an adversary first erratically selects $n$ local convex loss functions $\{ {{l_{i,t}}:\mathbb{X} \to \mathbb{R}} \}$ and $n$ local convex constraint functions $\{ {{c_{i,t}}:\mathbb{X} \to {\mathbb{R}^{{m_i}}}} \}$ for ${\rm{ }}i \in [n]$, where $\mathbb{X} \subseteq {\mathbb{R}^p}$ is known convex set to the agents, and both ${m_i}$ and $p$ are positive integers. After that, all agents exchange information with their neighbors. Then, the agents simultaneously select their local decisions $\{ {{x_{i,t}} \in \mathbb{X}} \}$ without prior access to $\{ {{l_{i,t}}} \}$ and $\{ {{c_{i,t}}} \}$. Moreover, the values of ${l_{i,t}}( {{x_{i,t}}} )$ and ${c_{i,t}}( {{x_{i,t}}} )$ are privately revealed to agent $i$. The objective is to choose the decision sequence $\{ {{x_{i,t}}} \}$ such that both network regret \par\nobreak\vspace{-10pt}
\begin{small}
\begin{flalign}
{\rm{Net}\mbox{-}\rm{Reg}}( {\{ {{x_{i,t}}} \},{y_{[ T ]}}} ) := \frac{1}{n}\sum\limits_{i = 1}^n {\sum\limits_{t = 1}^T {{l_t}( {{x_{i,t}}} )} }  - \sum\limits_{t = 1}^T {{l_t}( {{y_t}} )} \label{regret-eq1}
\end{flalign}
\end{small}%
and network CCV\par\nobreak\vspace{-10pt}
\begin{small}
\begin{flalign}
{\rm{Net}\mbox{-}\rm{CCV}}( {\{ {{x_{i,t}}} \}} ) := \frac{1}{n}\sum\limits_{i = 1}^n {\sum\limits_{t = 1}^T {\| {{{[ {{c_t}( {{x_{i,t}}} )} ]}_ + }} \|} } \label{CCV-eq2}
\end{flalign}
\end{small}%
increase in a sublinear manner, where ${y_{[T]}} = ( {{y_1}, \cdot  \cdot  \cdot ,{y_T}} )$ is a benchmark, ${l_t}( x ) = \frac{1}{n}\sum\nolimits_{j = 1}^n {{l_{j,t}}( x )} $ and ${c_t}( x ) = {\rm{col}}\big( {{c_{1,t}}( x ), \cdot  \cdot  \cdot ,{c_{n,t}}( x )} \big) \in {\mathbb{R}^m}$ are the global loss and constraint functions at iteration $t$, respectively, and $m = \sum\nolimits_{i = 1}^n {{m_i}} $.

In this paper, we consider dynamic and static network regret, i.e., ${\rm{Net}\mbox{-}\rm{Reg}}( {\{ {{x_{i,t}}} \},{\check{x}_{[ T ]}^ *}} )$ and ${\rm{Net}\mbox{-}\rm{Reg}}( {\{ {{x_{i,t}}} \},{\hat x_{[ T ]}^*}} )$. For dynamic network regret, $\check{x}_{[ T ]}^ *  = ( {\check{x}_1^ * , \cdots,\check{x}_T^ * } )$ is the optimal decision sequence, where $\check{x}_t^ * \in \mathbb{X}$ is the minimizer of ${l_t}( x )$ subject to ${c_t}( x ) \le {\mathbf{0}_m}$ for $t \in [T]$.
For static network regret, $\hat x_{[ T ]}^ *  = ( {\hat{x}^ * , \cdots,\hat{x}^ * } )$ is the optimal static decision sequence, where $\hat{x}^ * \in \mathbb{X}$ is the minimizer of $\sum\nolimits_{t = 1}^T {{l_t}( x )} $ subject to ${c_t}( x ) \le {\mathbf{0}_m}$ for $t \in [T]$.

Throughout this paper, the following common assumptions are made for the local loss and constraint functions of the agents, and the graph $\mathcal{G}_t$, see, e.g., \cite{Yuan2022, Yi2023}.
Moreover, the definition of strongly convex functions is given.
\begin{assumption}
\;\;(a) Centered at the origin, the convex set $\mathbb{X}$ contains the ball of radius $r( \mathbb{X} )$ and is contained in the ball of radius $R( \mathbb{X} )$, i.e.,\par\nobreak\vspace{-10pt}
\begin{small}
\begin{flalign}
r( \mathbb{X} ){\mathbb{B}^p} \subseteq \mathbb{X} \subseteq R( \mathbb{X} ){\mathbb{B}^p}. \label{ass-eq1}
\end{flalign}
\end{small}
\;\;(b) For any $x \in \mathbb{X}$, there exists a positive constant ${F_1}$ such that the convex local loss function ${l_{i,t}}$ and constraint function~${c_{i,t}}$ for any $i \in [n]$, $t \in {\mathbb{N}_ + }$ satisfy\par\nobreak\vspace{-10pt}
\begin{small}
\begin{subequations}
\begin{flalign}
&| {{l_{i,t}}( x )} | \le {F_1}, \label{ass-eq2a}\\
&\| {{c_{i,t}}( x )} \| \le {F_1}. \label{ass-eq2b}
\end{flalign}
\end{subequations}
\end{small}
\;\;(c) For any $i \in [n]$, $t \in {\mathbb{N}_ + }$, the subgradients $\partial {f_{i,t}}( x )$ and $\partial {g_{i,t}}( x )$ exist, moreover, for any $x \in \mathbb{X}$, there exists a positive constant ${F_2}$ such that\par\nobreak\vspace{-10pt}
\begin{small}
\begin{subequations}
\begin{flalign}
&\| {\partial {l_{i,t}}( x )} \| \le {F_2}, \label{ass-eq3a}\\
&\| {\partial {c_{i,t}}( x )} \| \le {F_2}. \label{ass-eq3b}
\end{flalign}
\end{subequations}
\end{small}
\end{assumption}
\begin{assumption}
The graph $\mathcal{G}_t$ for any $t \in {\mathbb{N}_ + }$ satisfies that

\noindent
\;\;(a) its mixing matrix ${W_t}$ is doubly stochastic, i.e., for any $i,j \in [ n ]$, ${\sum\nolimits_{i = 1}^n {[ {{W_t}} ]} _{ij}} = {\sum\nolimits_{j = 1}^n {[ {{W_t}} ]} _{ij}} = 1$;

\noindent
\;\;(b) ${[ {{W_t}} ]_{ij}} \ge \omega$ if $( {j,i} ) \in {\mathcal{E}_t}$ or $i = j$, where $\omega  \in ( {0,1} )$;

\noindent
\;\;(c) the time-varying directed graph $( {\mathcal{V}, \cup _{s = 0}^{B - 1}{\mathcal{E} _{t + s}}} )$ is strongly connected, where integer $B > 0$.
\end{assumption}
\begin{definition}
For any $i \in [n]$, $t \in {\mathbb{N}_ + }$, the local loss function $ {l_{i,t}} $ is strongly convex with the parameter $\mu  > 0$, i.e., for any $x,y \in \mathbb{X}$, \par\nobreak\vspace{-10pt}
\begin{small}
\begin{flalign}
{l_{i,t}}( x ) \ge {l_{i,t}}( y ) + \langle {x - y,\partial {l_{i,t}}( y )} \rangle  + \frac{\mu }{2}\| {x - y} \|^2. \label{def-eq1}
\end{flalign}
\end{small}%
\end{definition}

\section{Proposed Distributed Algorithm and Performance Analysis}
This section proposes a distributed bandit online primal--dual projection algorithm with one-point sampling. In addition, we establish network regret and CCV bounds for convex and strongly convex local loss functions, respectively.

\subsection{Proposed Distributed Algorithm}
In the considered DBCO problem, for $i \in [ n ]$, $t \in [ T ]$, agent~$i$ only knows the values of its convex local loss function ${l_{i,t}}$ and constraint function ${c_{i,t}}$ at $x_{i,t}$. Inspired by the one-point stochastic gradient approximations in \cite{Flaxman2005, Yi2021b}, the subgradients $\partial {l_{i,t}}( {{e_{i,t}}} )$ and $\partial {[ {{c_{i,t}}( {{e_{i,t}}} )} ]_ + }$ are estimated by\par\nobreak\vspace{-10pt}
\begin{small}
\begin{flalign}
\hat \partial {l_{i,t}}( {{e_{i,t}}} ) := \frac{p}{{{\delta _t}}}\big( {{l_{i,t}}( {{e_{i,t}} + {\delta _t}{u_{i,t}}} )} \big){u_{i,t}} \in {\mathbb{R}^p} \label{algorithm-eq1}
\end{flalign}
\end{small}%
and\par\nobreak\vspace{-10pt}
\begin{small}
\begin{flalign}
\hspace{-7pt}
\hat \partial {[ {{c_{i,t}}( {{e_{i,t}}} )} ]_ + } := \frac{p}{{{\delta _t}}}{\big( {{{[ {{c_{i,t}}( {{e_{i,t}} + {\delta _t}{u_{i,t}}} )} ]}_ + }} \big)^T} \otimes {u_{i,t}} \in {{\mathbb{R}}^{p \times {m_i}}}, \label{algorithm-eq2}
\end{flalign}
\end{small}%
respectively, where ${e_{i,t}} = {x_{i,t}} - {\delta _t}{u_{i,t}} \in {\mathbb{R}^p}$, ${\delta _t} \in \big( {0,r( \mathbb{X} ){\xi _t }} \big]$ is an exploration parameter, $r( \mathbb{X} )$ is a known constant given in Assumption~1, ${\xi _t} \in ( {0,1} ]$ is a shrinkage coefficient, and ${u_{i,t}} \in {\mathbb{S}^p}$ is a uniformly distributed random vector.

Note that the estimators $\hat \partial {l_{i,t}}$ and $\hat \partial {[ {{c_{i,t}}} ]_ + }$ are unbiased subgradient estimators of ${{\hat l}_{i,t}}$ and ${[ {{{\hat c}_{i,t}}} ]_ + }$, which are uniformly smoothed versions of ${l_{i,t}}$ and ${[ {{c_{i,t}}} ]_ + }$, and respectively defined as\par\nobreak\vspace{-10pt}
\begin{small}
\begin{flalign}
{{\hat l}_{i,t}}( x ) := {\mathbf{E}_{v \in {\mathbb{B}^p}}}[ {{l_{i,t}}( {x + \delta_t v} )} ],\forall x \in ( {1 - \xi_t } )\mathbb{X} \label{algorithm-eq3}
\end{flalign}
\end{small}%
and\par\nobreak\vspace{-10pt}
\begin{small}
\begin{flalign}
{[ {{{\hat c}_{i,t}}( x )} ]_ + } := {\mathbf{E}_{v \in {\mathbb{B}^p}}}{[ {{{\hat c}_{i,t}}( {x + \delta_t v} )} ]_ + },\forall x \in ( {1 - \xi_t } )\mathbb{X}, \label{algorithm-eq4}
\end{flalign}
\end{small}%
where $\mathbf{E}$ denotes the expectation, and $v$ is chosen uniformly at random.

By integrating the one-point stochastic gradient approximations \eqref{algorithm-eq1} and \eqref{algorithm-eq2} with the algorithm with two-point sampling in \cite{Yi2023}, our algorithm is proposed in pseudo-code as Algorithm~1. Specifically, the updating rule of the local primal decision variables is presented as \eqref{Algorithm1-eq1}--\eqref{Algorithm1-eq3}. Moreover, the updating direction of the local primal decision variables is calculated based on subgradient estimators \eqref{algorithm-eq1} and \eqref{algorithm-eq2}. In addition, the updating rule of the dual decision variables is given as \eqref{Algorithm1-eq4} and \eqref{Algorithm1-eq5}.

%\KP{Note that different from Algorithm~1 in \cite{Yi2022} which uses full-information feedback for local loss and constraint functions, we use one-point sampling for those. The resulting challenge is that we can not directly take derivative calculation to obtain the subgradients of local loss and constraint functions. To cope with this challenge, we use one-point stochastic gradient approximations in the proposed algorithm.}

\begin{algorithm}[t]
  \caption{Distributed Bandit Online Primal--Dual Projection Algorithm with One-Point Sampling} % 名称
  \begin{algorithmic}
  \renewcommand{\algorithmicrequire}{\textbf{Input:}}
  \REQUIRE
     non-increasing sequences $\{ {{\alpha _t}} \}$, $\{ {{\beta _t}} \}$, $\{ {{\gamma _t}} \} \subseteq ( {0, + \infty } )$, $\{ {{\xi _t}} \} \subseteq ( {0,1} )$ and $\{ {{\delta _t}} \} \subseteq \big( {0,r( \mathbb{X} ){\xi _t}} \big]$ for $t \in [ T ]$.
  \renewcommand{\algorithmicrequire}{\textbf{Initialize:}}
  \REQUIRE
     ${e_{i,1}} \in ( {1 - {\xi _1}} )\mathbb{X}$, ${u_{i,1}}$, ${x_{i,1}} = {e_{i,1}} + {\delta _1}{u_{i,1}}$ and ${q_{i,1}} = {\mathbf{0}_{{m_i}}}$ for $i \in [ n ]$.
    \FOR {$t = 1, \cdot  \cdot  \cdot, T-1 $}
    \FOR {$i = 1,\cdot  \cdot  \cdot,n$ in parallel}
    \STATE Broadcast ${e_{i,t}}$ and receive ${e_{j,t}}$ via the graph $\mathcal{G}_t$;
    \STATE Select vector ${u_{i,t}}$ independently and uniformly at random;
    \STATE Sample ${l_{i,t}}( {{x_{i,t}}} )$ and ${c_{i,t}}( {{x_{i,t}}} )$;
    \STATE Update
      \begin{flalign}
       {z_{i,t + 1}} &= \sum\limits_{j = 1}^n {{{[ {{W_t}} ]}_{ij}}{e_{j,t}}}, \label{Algorithm1-eq1}\\
       {\hat \omega _{i,t + 1}} &= \hat \partial {l_{i,t}}( {{e_{i,t}}} ) + \hat \partial {[ {{c_{i,t}}( {{e_{i,t}}} )} ]_ + }{q_{i,t}}, \label{Algorithm1-eq2}\\
       {e_{i,t + 1}} &= {\mathcal{P}_{( {1 - {\xi _{t + 1}}} )\mathbb{X}}}( {{z_{i,t + 1}} - {\alpha _{t + 1}}{{\hat \omega }_{i,t + 1}}} ), \label{Algorithm1-eq3}\\
       {x_{i,t + 1}} &= {e_{i,t+1}} + {\delta _{t+1}}{u_{i,t+1}}, \label{Algorithm1-eq4}\\
       {q_{i,t + 1}} &= \big[ ( {1 - {\beta _{t + 1}}{\gamma _{t + 1}}} ){q_{i,t}} + {\gamma _{t + 1}}{{[ {{c_{i,t}}( {{x_{i,t}}} )} ]}_ + } \big]_ + . \label{Algorithm1-eq5}
      \end{flalign}
    \ENDFOR
    \ENDFOR
  \renewcommand{\algorithmicensure}{\textbf{Output:}}
  \ENSURE
      $\{ {{x_{i,t}}} \}$ for $i \in [ n ]$ and $t \in [ T ]$.
  \end{algorithmic}
\end{algorithm}

\subsection{Network Regret and CCV}
For convex local loss functions, we first derive dynamic network regret and network CCV bounds as follows.
\begin{theorem}\label{thm1}
Let Assumptions 1 and 2 hold. For all $i \in [ n ]$, let $\{ {{x_{i,t}}} \}$ be the sequence induced by Algorithm~1 with\par\nobreak\vspace{-10pt}
\begin{small}
\begin{flalign}
\nonumber
&{\alpha _t} = \frac{{r{{( \mathbb{X} )}^2}}}{{20{p^2}F_1^2{{( {t + 1} )}^{{g _1}}}}}, {\beta _t} = \frac{2}{{{t^{{g _2}}}}}, {\gamma _t} = \frac{1}{{{t^{1 - {g _2}}}}}, \\
&{\xi _t} = \frac{1}{{{{( {t + 1} )}^{{g _3}}}}}, {\delta _t} = \frac{{r( \mathbb{X} )}}{{{{( {t + 1} )}^{{g _3}}}}},t \in {\mathbb{N}_ + }, \label{theorem1-eq1}
\end{flalign}
\end{small}%
where the constants ${g _1} \in ( {0, 1} )$, ${g _2} \in ( {0,{g _1}/4} )$ and ${g _3} \in ( {{g _2},({g _1} - 2{g _2})/2} )$. Then, for any benchmark ${y_{[ T ]}} \in {\mathcal{X}_T}$ and the total number of iterations $T \in {\mathbb{N}_ + }$,\par\nobreak\vspace{-10pt}
\begin{small}
\begin{flalign}
\nonumber
&\mathbf{E}[ {{\rm{Net}\mbox{-}\rm{Reg}}( {\{ {{x_{i,t}}} \},{y_{[ T ]}}} )} ] \\
&= {\cal O}( {{T^{\max \{ {{g _1},1 - {g _1} + 2{g _2} + 2{g _3},1 + {g _2} - {g _3}} \}}} + {T^{{g _1}}}{P_T}} ), \label{theorem1-eq2}\\
&{\mathbf{E}}[ {{\rm{Net} \mbox{-} \rm{CCV}}( {\{ {{x_{i,t}}} \}} )} ]
= {\cal O}( {{T^{1 - {g _2}/2}}} ), \label{theorem1-eq3}
\end{flalign}
\end{small}%
where ${P_T} = \sum\nolimits_{t = 1}^{T - 1} {\| {{y_{t + 1}} - {y_t}} \|} $ is the path--length of ${y_{[ T ]}}$.
\end{theorem}
%\begin{proof}
%The explicit expressions of the right-hand sides of (26) and (27), and
The proof is given in Appendix B.
%\end{proof}
\begin{remark}\label{rem1}
%\noindent \textbf{Remark 1}.
From \eqref{theorem1-eq2} and \eqref{theorem1-eq3}, we know that Algorithm~1 achieves sublinear dynamic network regret and network CCV if the path--length ${P_T}$ also increases in a sublinear manner. Compared to the algorithms with one-point sampling in \cite{Flaxman2005, Saha2011, Chen2019, Yuan2022}, our Algorithm~1 has an advantage that does not use the total number of iterations~$T$ as shown in (18). Moreover, our Algorithm~1 is distributed, while the algorithms in \cite{Flaxman2005, Saha2011, Chen2019} are centralized. More importantly, our Algorithm~1 uses one-point bandit feedback for local constraint functions, while \cite{Flaxman2005, Saha2011} do not consider inequality constraints, and the algorithms in \cite{Chen2019, Yuan2022} use full-information feedback for local constraint functions.
\end{remark}
\begin{remark}\label{rem2}
Intuitively, our Algorithm~1 can be viewed as a one-point bandit feedback version of Algorithm~2 in \cite{Yi2023}. However, the proof of our Theorem~\ref{thm1} has nontrivial differences compared to the proof of Theorem~3 for Algorithm~2 with two-point sampling in \cite{Yi2023} due to the following reasons.
Firstly, different from the uniformly bounded two-point stochastic gradient approximation, the one-point stochastic gradient approximation is inversely proportional to the exploration parameter $\delta _t$ as in Lemma~2.
This subtle yet critical difference poses a significant challenge since the exploration parameter ${{\delta _t}}$ is time-varying.
It is worth mentioning that the method in \cite{Yi2023} cannot be directly applied to derive the result of Lemma~5 in this paper. To overcome this dilemma, we reanalyse all formulas related to the time-varying exploration parameter~${{\delta _t}}$ to establish network regret and CCV bounds. That also causes the designed parameters of our Algorithm~1 to be different from those of Algorithm~2 in \cite{Yi2023}. Secondly, the updating rules of the local dual decision variables of our Algorithm~1 is different from that of Algorithm~2 in \cite{Yi2023}, i.e., our Algorithm~1 does not use ${\big( {\partial {{[ {{c_{i,t}}( {{e_{i,t}}} )} ]}_ + }} \big)^T}( {{e_{i,t + 1}} - {e_{i,t}}} )$. Therefore, we also need to reanalyse the results on the local dual variables, see Lemma~4.
Additionally, note that different from the algorithm with one-point sampling in \cite{Yuan2022} which uses one-point sampling for the time-varying local loss functions but full-information feedback for the fixed local constraint functions, our Algorithm~1 use one-point sampling for both the time-varying local loss and constraint functions.
Furthermore, the update rules of these two algorithms are substantially different.
\end{remark}

Next, we establish static network regret and CCV bounds
by replacing the benchmark ${y_{[ T ]}}$ with \ $\hat x_{[ T ]}^*$ as follows.

\begin{corollary}\label{cor1}
Let Assumptions~1 and 2 hold. For all $i \in [ n ]$, let $\{ {{x_{i,t}}} \}$ be the sequence induced by Algorithm~1 with (18), where the constants ${g _1} = g  + 3/4$, ${g _2} = g $, ${g _3} = 1/4$, and $g  \in ( {0, 1/4} )$ can be arbitrarily chosen. Then, for any the total number of iterations $T \in {\mathbb{N}_ + }$,\par\nobreak\vspace{-10pt}
\begin{small}
\begin{flalign}
&\mathbf{E}[{\rm{Net}\mbox{-}\rm{Reg}}( {\{ {{x_{i,t}}} \},\hat{x}_{[ T ]}^*})] = \mathcal{O}( {{T^{3/4 + g}}} ), \label{corollary1-eq1}\\
&{\mathbf{E}}[ {{\rm{Net} \mbox{-} \rm{CCV}}( {\{ {{x_{i,t}}} \}} )} ]
= \mathcal{O}({T^{1 - {g}/2}}).\label{corollary1-eq2}
\end{flalign}
\end{small}
\end{corollary}
\begin{remark}\label{rem3}
%\noindent \textbf{Remark 3}.
From \eqref{corollary1-eq1} and \eqref{corollary1-eq2}, we know that Algorithm~1 achieves sublinear static network regret and CCV. As shown in TABLE~1, an $\mathcal{O}({T^{3/4}})$ static network regret and the $\mathcal{O}({T^{5/8}})$ network CCV are established by the algorithm with one-point sampling in \cite{Yuan2022} when we choose $g = 3/4$, which are smaller than the bounds in \eqref{corollary1-eq1} and \eqref{corollary1-eq2}. That is reasonable since the algorithm in \cite{Yuan2022} uses full-information feedback for local constraint functions while our Algorithm~1 only uses one-point sampling for local constraint functions. In addition, in the absence of inequality constraints, our Algorithm 1 achieves an $\mathcal{O}({T^{3/4}})$ static network regret bound, which is same as that established by the algorithms with one-point samping in \cite{Flaxman2005, Chen2019, Yuan2022}. However, the algorithms in \cite{Flaxman2005, Chen2019, Yuan2022} additionally use the total number of iterations~$T$, and the algorithms in \cite{Flaxman2005, Chen2019} are centralized. Although the bound is larger than that established by the algorithm with two-point samping in \cite{Yi2023}, our Algorithm~1 requires less sampling, computational, and memory requirements.

\end{remark}
For strongly convex local loss functions, reduced dynamic network regret and network CCV bounds are established in the following theorem.
\begin{theorem}\label{thm2}
Let Assumptions~1 and 2 hold. For all $i \in [ n ]$, let $\{ {{x_{i,t}}} \}$ be the sequence induced by Algorithm~1 with \eqref{theorem1-eq1}, where the constants ${g _1} \in ( {0, 1} )$, ${g _2} \in ( {0,{g _1}/4} )$ and ${g _3} \in ( {{g _2},({g _1} - 2{g _2})/2} )$. Then, for any benchmark ${y_{[ T ]}} \in {\mathcal{X}_T}$ and the total number of iterations $T \in {\mathbb{N}_ + }$,\par\nobreak\vspace{-10pt}
\begin{small}
\begin{flalign}
\nonumber
&\mathbf{E}[ {{\rm{Net}\mbox{-}\rm{Reg}}( {\{ {{x_{i,t}}} \},{y_{[ T ]}}} )} ] \\
&= {\cal O}( {{T^{\max \{ {1 - {g _1} + 2{g _2} + 2{g _3},1 + {g _2} - {g _3}} \}}} + {T^{{g _1}}}{P_T}} ), \label{theorem2-eq1}\\
&{\mathbf{E}}[ {{\rm{Net} \mbox{-} \rm{CCV}}( {\{ {{x_{i,t}}} \}} )} ]
= {\cal O}( {{T^{1 - {g _2}/2}}} ). \label{theorem2-eq2}
\end{flalign}
\end{small}
\end{theorem}
%\begin{proof}
The proof is given in Appendix C.
%\end{proof}
\begin{remark}\label{rem4}
%\noindent \textbf{Remark 4}.
From \eqref{theorem2-eq1} and \eqref{theorem2-eq2}, we know that sublinear dynamic network regret and network CCV bounds are established if ${P_T}$ also increases sublinearly. Moreover, when ${g _1} > 1 + {g _2} - {g _3} \ge 1 - {g_1} + 2{g_2} + 2{g_3}$, compared to the dynamic network regret bound $\mathcal{O}( {{T^{{g _1}}} + {T^{{g _1}}}{P_T}} )$ established in Theorem~1, the dynamic network regret bound $\mathcal{O}( {{T^{1 + {g _2} - {g _3}}} + {T^{{g _1}}}{P_T}})$ in Theorem~2 is smaller.
\end{remark}

Next, compared with Corollary~1, we establish reduced static network regret bound
by replacing the benchmark ${y_{[ T ]}}$ in Theorem~2 with \ $\hat x_{[ T ]}^*$ in the following corollary.
\begin{corollary}\label{cor2}
Let Assumptions~1 and 2 hold. For all $i \in [ n ]$, let $\{ {{x_{i,t}}} \}$ be the sequence induced by Algorithm~1 with\par\nobreak\vspace{-10pt}
\begin{small}
\begin{flalign}
\nonumber
&{\alpha _t} = \frac{{r{{( \mathbb{X} )}^2}}}{{20{p^2}F_1^2{{( {t + 1} )}}}}, {\beta _t} = \frac{2}{{{t^{{g _1}}}}}, {\gamma _t} = \frac{1}{{{t^{1 - {g _1}}}}},\\
&{\xi _t} = \frac{1}{{{{( {t + 1} )}^{{g _2}}}}}, {\delta _t} = \frac{{r( \mathbb{X} )}}{{{{( {t + 1} )}^{{g _2}}}}},t \in {\mathbb{N}_ + }, \label{corollary2-eq1}
\end{flalign}
\end{small}%
where the constants ${g _1} = g $, ${g _2} = ( {1 - g } )/3$ and $g  \in ( {0,1/4} )$. Then, for the total number of iterations $T \in {\mathbb{N}_ + }$,\par\nobreak\vspace{-10pt}
\begin{small}
\begin{flalign}
&\mathbf{E}[{\rm{Net}\mbox{-}\rm{Reg}}( {\{ {{x_{i,t}}} \},\hat{x}_{[ T ]}^*} )] = \mathcal{O}( {{T^{2/3 + 4g /3}}} ), \label{corollary2-eq2}\\
&{\mathbf{E}}[ {{\rm{Net} \mbox{-} \rm{CCV}}( {\{ {{x_{i,t}}} \}} )} ]
= {\cal O}( {{T^{1 - {g}/2}}} ). \label{corollary2-eq3}
\end{flalign}
\end{small}
%where ${\varpi _5} = \lceil {{\varpi _4} - 1} \rceil $ with $\lceil  \cdot  \rceil $ being the ceiling function, and ${\varpi _4} = {( {r{{( \mathbb{X} )}^2}/360{p^2}F_2^2\mu } )^{1/( {1 - {\theta _1}} )}}$.
\end{corollary}

%\begin{proof}
%The explicit expressions of the right-hand sides of (32) and (33), and
%The proof is given in Appendix C.
%\end{proof}
\begin{remark}\label{rem5}
%\noindent \textbf{Remark 5}.
Note that different from Theorem~2, in Corollary~2 we can choose ${\alpha _t} = r{( \mathbb{X} )^2}/20{p^2}F_1^2( {t + 1} )$ for establishing static network regret bound.
To the best of our knowledge, Corollary~2 is among the first to establish such static network regret and CCV bounds as in \eqref{corollary2-eq2} and \eqref{corollary2-eq3} for constrained DBCO with one-point sampling.
That is different from those established by the algorithm with one-point sampling in \cite{Yuan2022} as presented in TABLE~1 due to the trade-off parameter $g$. Moreover, unlike the algorithm in \cite{Yuan2022}, our Algorithm~1 does not need to know the strongly convex parameter $\mu$.
Note that ${T^{2/3 + 4g/3}} < {T^{3/4 + g}}$ when $g  \in ( {0,1/4} )$, i.e., the static network regret bound established in Corollary~2 is smaller than that established in Corollary~1. In the absence of inequality constraints, our Algorithm~1 establishes an $\mathcal{O}( {{T^{2/3}}} )$ static network regret bound, which is the same as that established by the centralized algorithm with one-point sampling in \cite{Saha2011}. However, the algorithm in \cite{Saha2011} needs to know the total number of iterations $T$ in advance, and does not consider inequality constraints. The static network regret bound \eqref{corollary2-eq2} is larger than that established by the distributed algorithm with two-point sampling in \cite{Yi2023}. That is reasonable since our Algorithm~1 uses less information. The network CCV bound is larger than that established by the distributed algorithm in \cite{Yuan2022}. That is also reasonable because the algorithm in \cite{Yuan2022} uses full-information feedback for local constraint functions, while our Algorithm~1 only uses one-point sampling.
\end{remark}

\section{Numerical Simulations}
To evaluate the performance of Algorithm~1, this section considers a distributed online ridge regression problem with time-varying linear inequality constraints. At iteration $t$, the local loss function is ${l_{i,t}}( x ) = \frac{1}{2}{( {a_{i,t}^Tx - {m_{i,t}}} )^2} + \lambda {\| x \|^2}$ and the local constraint function is ${c_{i,t}}\left( x \right) = {B_{i,t}}x - {b_{i,t}}$, where ${a_{i,t}} \in {\mathbb{R}^{16}}$ is the feature information and its elements are randomly generated from the uniform distribution in $[ { - 5,5} ]$,
${m_{i,t}} \in \mathbb{R}$ is the label information and satisfies ${m_{i,t}} = a_{i,t}^T{x_{i,1}} + \frac{{{\upsilon _{i,t}}}}{{4t}}$ with ${{\upsilon _{i,t}}}$ being randomly generated from the uniform distribution in $[ { 0,1} ]$,
the parameter $\lambda  = 5 \times {10^{ - 6}}$, and the elements of ${B_{i,t}} \in {\mathbb{R}^{{2} \times {16}}}$ and ${b_{i,t}} \in {\mathbb{R}^{{2}}}$ are randomly generated from the uniform distribution in $[ {0,2} ]$ and $[ {0,1} ]$, respectively.
The network with $100$ agents is modeled by a random undirected connected graph. Specifically, connections
between agents are random and the probability of two agents being connected is $0.1$.
To guarantee that Assumption~2 holds, we add edges $( {i,i + 1} )$ for $i \in [ {n - 1} ]$, and let ${[ {{W_t}} ]_{ij}} = 1/n$ if $( {j,i} ) \in \mathcal{E}_t$ and ${[ {{W_t}} ]_{ii}} = 1 - \sum\nolimits_{j = 1}^n {{{[ {{W_t}} ]}_{ij}}}$ with $j \ne i$. Moreover, we set $\mathbb{X} = {[ { - 2,2} ]^{16}}$.

We compare our Algorithm~1 with the distributed bandit online algorithm in \cite{Yuan2022} where one-point sampling is used for the local loss function and full-information feedback is used for the local constraint function, and the distributed bandit online algorithm in \cite{Yi2023} where two-point sampling is used for the local loss and constraint functions.
Note that our Algorithm~1 uses one-point sampling for the local loss and constraint functions, which requires less information compare to the algorithms in \cite{Yuan2022, Yi2023}.
To demonstrate this point, the evolutions of the total number of samples of our Algorithm~1 and the algorithm in \cite{Yi2023} are illustrated in Fig.~1.
Since the algorithm in \cite{Yuan2022} uses full-information feedback for the local constraint function, we do not count its total number of samples for the sake of fairness. 
Fig.~1 demonstrates that the total number of samples of our Algorithm~1 is significantly less than that of the algorithm in \cite{Yi2023}.
Fig.~2 and Fig.~3 illustrate the evolutions of the average cumulative loss $\frac{1}{n}\sum\nolimits_{i = 1}^n {\sum\nolimits_{t = 1}^T {{l_t}( {{x_{i,t}}} )} } /T$ and the average CCV $\frac{1}{n}\sum\nolimits_{i = 1}^n {\sum\nolimits_{t = 1}^T {\| {{{[ {{c_t}( {{x_{i,t}}} )} ]}_ + }} \|} } /T$, respectively.
Fig.~2 demonstrates that the average cumulative loss of our Algorithm~1 is almost the same as that of the algorithm in \cite{Yuan2022}, and slightly larger than that of the algorithm in \cite{Yi2023}.
This is reasonable because the algorithm in \cite{Yuan2022} uses one-point sampling for the local loss function, which is also used by our Algorithm~1, while the algorithm in \cite{Yi2023} uses two-point sampling for the local loss function.
Fig.~3 demonstrates that the average CCV of our Algorithm~1 is larger than that of the algorithms in \cite{Yuan2022, Yi2023}.
This is also reasonable because the algorithm in \cite{Yuan2022} uses full-information feedback for the local constraint function, and the algorithm in \cite{Yi2023} uses two-point sampling, while our Algorithm~1 uses one-point sampling.
These results are in consistent with the theoretical results established in Corollary~1.

\begin{figure}[!ht]
 \centering
  \includegraphics[width=7.8cm]{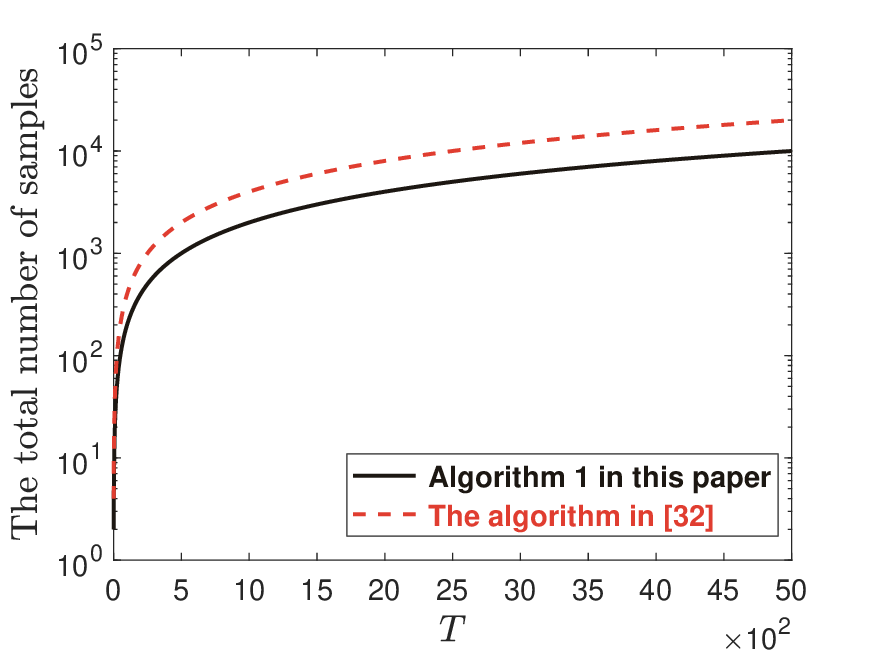}
  \caption{Evolutions of the total number of samples.}
\end{figure}

\begin{figure}[!ht]
 \centering
  \includegraphics[width=7.8cm]{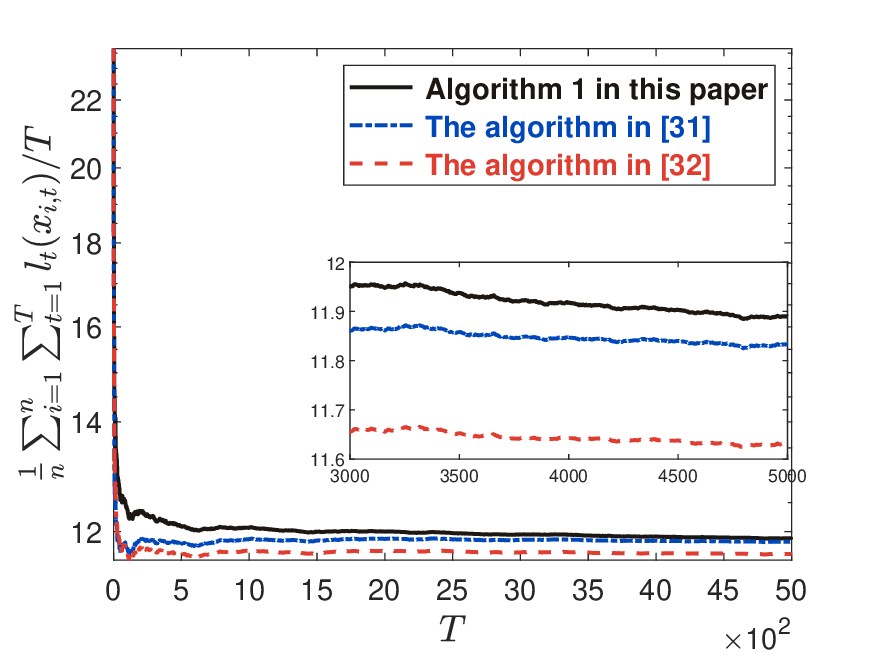}
  \caption{Evolutions of $\frac{1}{n}\sum\nolimits_{i = 1}^n {\sum\nolimits_{t = 1}^T {{l_t}( {{x_{i,t}}} )} } /T$.}
\end{figure}

\begin{figure}[!ht]
 \centering
  \includegraphics[width=7.8cm]{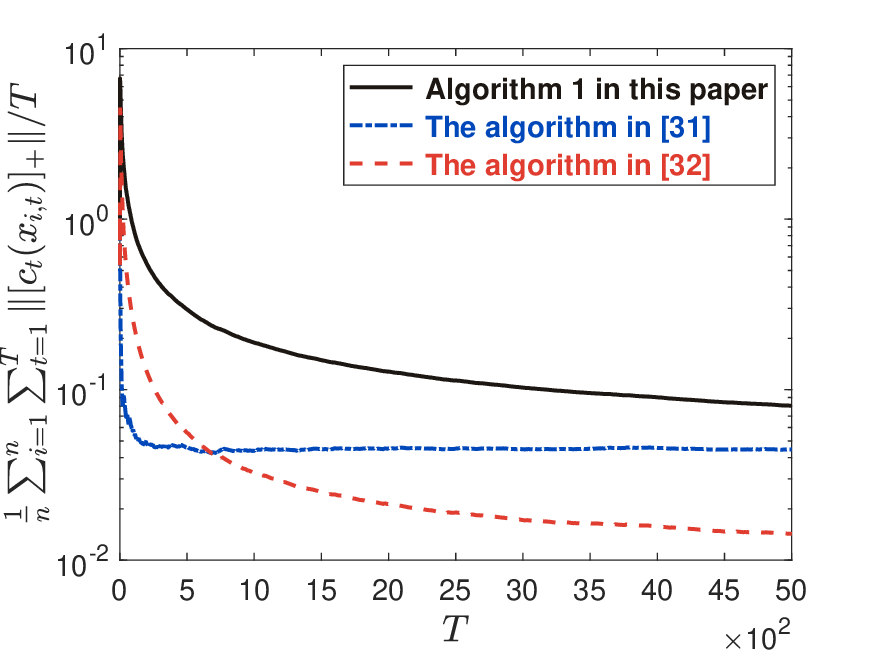}
  \caption{Evolutions of $\frac{1}{n}\sum\nolimits_{i = 1}^n {\sum\nolimits_{t = 1}^T {\| {{{[ {{c_t}( {{x_{i,t}}} )} ]}_ + }} \|} } /T$.}
\end{figure}

\section{Conclusions}
This paper considered the DBCO problem with time-varying constraints. We proposed a distributed bandit online primal--dual algorithm with one-point sampling where one-point sampling for local loss and constraint functions was used. We established sublinear network regret and CCV bounds for both convex and strongly convex local loss functions. In the future, we will investigate distributed bandit online projection algorithms with compressed communication to reduce communication overhead.

\section*{Acknowledgements}
The first author would like to thank Prof. Ming Cao for the useful discussion, who is affiliated with the Engineering and Technology Institute Groningen, Faculty of Science and Engineering, University of Groningen, 9747 AG Groningen, The Netherlands.

\appendix

\vspace{2mm}
%\hspace{-3mm}\emph{A. Useful Lemmas}

%-----------------------------引理1--------------------------------------------------------------
\subsection{Useful Lemmas}

%\TY{Do we need to add something like in order to xxx, the following lemmas are useful. Look at Xinlei's paper.}

%\begin{lemma}\label{lem1}
%Let ${W_t}$ be the adjacency matrix associated with a time-varying graph satisfying Assumption 5. Then,\par\nobreak\vspace{-10pt}
%\begin{small}
%\begin{flalign}
%\left| {{{\left[ {\Psi _s^t} \right]}_{ij}} - \frac{1}{n}} \right| \le \tau {\lambda ^{t - s}},\forall i,j \in \left[ n \right],\forall t \ge s \ge 1,
%\end{flalign}
%\end{small}%
%where $\Psi _s^t = {W_t}{W_{t - 1}} \cdot  \cdot  \cdot {W_s}$, $\tau  = {\left( {1 - \omega /4{n^2}} \right)^{ - 2}} > 1$, and $\lambda  = {\left( {1 - \omega /4{n^2}} \right)^{1/B}} \in \left( {0,1} \right)$.
%\end{lemma}
\begin{lemma}\label{lem1}
(Lemma 3 in \cite{Yi2021b}) For convex set $\mathbb{K} \in {\mathbb{R}^p}$, $\tilde{q}$, $\tilde{p}$, and $\tilde{w}$ $ \in {\mathbb{R}^p}$, the results as follows hold.

\noindent
(a) When $\tilde{q} \le \tilde{p}$, one has\par\nobreak\vspace{-10pt}
\begin{small}
\begin{flalign}
\| {{{[ \tilde{q} ]}_ + }} \| \le \| \tilde{p} \|{{\; \mathrm{and} \;}}{[ \tilde{q} ]_ + } \le {[ \tilde{p} ]_ + }. \label{lemma1-eq1}
\end{flalign}
\end{small}%

\noindent
(b) When ${{\theta _1} } = \mathcal{P}{}_\mathbb{K}( {\tilde{w} - \tilde{q}} )$, one has\par\nobreak\vspace{-10pt}
\begin{small}
\begin{flalign}
\nonumber
&\;\;\;\;\; 2\langle {{\theta _1} - y_1,\tilde{q}} \rangle \\
&\le {\| {\theta _2 - \tilde{w}} \|^2} - {\| {\theta _2 - {\theta _1}} \|^2} - {\| {{\theta _1} - \tilde{w}} \|^2},\forall \theta _2 \in \mathbb{K}. \label{lemma1-eq2}
\end{flalign}
\end{small}%
\end{lemma}

\begin{lemma}\label{lem2}
Let Assumption 1 hold. For any $i \in [ n ]$, $t \in {\mathbb{N}_ + }$, $x \in ( {1 - {\xi _t}} )\mathbb{X}$, $q \in \mathbb{R}_ + ^{{m_i}}$, the functions ${\hat l_{i,t}}( x )$ and ${[ {{{\hat c}_{i,t}}( x )} ]_ + }$ are convex and satisfy\par\nobreak\vspace{-10pt}
\begin{small}
\begin{subequations}
\begin{flalign}
&\partial {\hat l_{i,t}}( x ) = {{\rm{\mathbf{E}}}_{{\mathfrak{U}_t}}}[ {\hat \partial {l_{i,t}}( x )} ], \label{lemma2-eq1a}\\
&{l_{i,t}}( x ) \le {\hat l_{i,t}}( x ) \le {l_{i,t}}( x ) + {F_2}{\delta _t}, \label{lemma2-eq1b}\\
&\| {\hat \partial {l_{i,t}}( x )} \| \le \frac{{p{F_1}}}{{{\delta _t}}}, \label{lemma2-eq1c}\\
&\partial {[ {{{\hat c}_{i,t}}( x )} ]_ + } = {{\rm{\mathbf{E}}}_{{\mathfrak{U}_t}}}\big[ {\hat \partial {{[ {{c_{i,t}}( x )} ]}_ + }} \big], \label{lemma2-eq1d}\\
\nonumber
&{q^T}{[ {{c_{i,t}}( x )} ]_ + } \le {q^T}{[ {{{\hat c}_{i,t}}( x )} ]_ + } \\
&\le {q^T}{[ {{c_{i,t}}( x )} ]_ + } + {F_2}{\delta _t}\| q \|, \label{lemma2-eq1e}\\
&\| {\hat \partial {{[ {{c_{i,t}}( x )} ]}_ + }} \| \le \frac{{p{F_1}}}{{{\delta _t}}}, \label{lemma2-eq1f}\\
&\| {{{[ {{{\hat c}_{i,t}}( x )} ]}_ + }} \| \le {F_1}. \label{lemma2-eq1g}
\end{flalign}
\end{subequations}
\end{small}%
\end{lemma}
\begin{proof}
\eqref{lemma2-eq1a}--\eqref{lemma2-eq1d} can be obtained from (39a)--(39d) of Lemma 6 in \cite{Yi2021b}, respectively. From (39e) of Lemma 6 in \cite{Yi2021b} and $q \in \mathbb{R}_ + ^{{m_i}}$, we have \eqref{lemma2-eq1e}. Replacing $| {{{[ {{c_{i,t}}( x )} ]}_j}} | \le {F_{{g_i}}}$ of Assumption~2 in \cite{Yi2021b} with \eqref{ass-eq2b}, similar to obtaining (39f) and (39g) of Lemma 6 in \cite{Yi2021b}, we have \eqref{lemma2-eq1f} and \eqref{lemma2-eq1g}.
\end{proof}

\subsection{The Proof of Theorem 1}
Denote ${\varpi _1} = \tau \sum\nolimits_{j = 1}^n {\| {{e_{j,1}}} \|} $, ${\varpi _2} = {\varpi _1}/\lambda ( {1 - \lambda } )$, ${\varpi _3} = 2{F_2} + {n^2}{F_2}{\tau ^2}/2{( {1 - \lambda } )^2}$, ${{\hat y}_t} = ( {1 - {\xi _t}} ){y_t}$, ${{\bar e}_t} = \frac{1}{n}\sum\nolimits_{j = 1}^n {{e_{j,t}}} $, $\psi _t^1 = R( \mathbb{X} ){\xi _t} + {\delta _t}$, $\psi _t^2 = 1/{\gamma _t} - 1/{\gamma _{t + 1}} + {\beta _{t + 1}}$, ${\psi _{i,t}} = {q_{i,t}} - {\mu _i}$, ${\Delta _{i,t}}( {{\mu_i}} ) = ( {1/2{\gamma _t}})\big( {{{\| \psi _{i,t} \|}^2} - ( {1 - {\beta _t}{\gamma _t}} ){{\| \psi _{i,t - 1} \|}^2}} \big)$ with ${{\mu _i}}$ being an arbitrary vector in $\mathbb{R}_ + ^{{m_i}}$, ${\varepsilon _{i,t}} = {e_{i,t}} - {e_{i,t + 1}}$, $\varepsilon _{i,t}^e = {e_{i,t + 1}} - {z_{i,t + 1}}$, $\varepsilon _{i,t}^1 = q_{i,t}^T\big( {{{[ {{c_{i,t}}( {{y_t}} )} ]}_ + } - {{[ {{c_{i,t}}( {{x_{i,t}}} )} ]}_ + }} \big)$, and $\varepsilon _{i,t}^2 = {\| {{{\hat y}_t} - {z_{i,t + 1}}} \|^2} - {\| {{{\hat y}_{t + 1}} - {z_{i,t + 2}}} \|^2} + {\| {{{\hat y}_{t + 1}} - {e_{i,t + 1}}} \|^2} - {\| {{{\hat y}_t} - {e_{i,t + 1}}} \|^2}$. Let ${\mathfrak{U}_t}$ be the $\sigma$-algebra induced by the independent and identically distributed variables ${u_{1,t}}, \cdot  \cdot  \cdot ,{u_{n,t}}$.

In the following, we give several preliminary lemmas.
The disagreement between the agents is first quantified.
\begin{lemma}\label{lem4}
(Lemma 4 in \cite{Yi2023}) Let Assumption~2 hold. For all $i \in [ n ]$ and $t \in {\mathbb{N}_ + }$, ${e_{i,t}}$ induced by Algorithm 1 satisfy\par\nobreak\vspace{-10pt}
\begin{small}
\begin{flalign}
\nonumber
\| {{e_{i,t}} - {{\bar e}_t}} \| &\le {\varpi _1}{\lambda ^{t - 2}} + \frac{1}{n}\sum\limits_{j = 1}^n {\| {\varepsilon _{j,t - 1}^e} \|}  + \| {\varepsilon _{i,t - 1}^e} \| \\
&\;\;\;\;\; + \tau \sum\limits_{k = 1}^{t - 2} {{\lambda ^{t - k - 2}}} \sum\limits_{j = 1}^n {\| {\varepsilon _{j,k}^e} \|}. \label{lemma3-eq1}
\end{flalign}
\end{small}%
where $\lambda  = {( {1 - \omega /4{n^2}} )^{1/B}} \in ( {0,1} )$ and $\tau  = {( {1 - \omega /4{n^2}} )^{ - 2}} > 1$.
\end{lemma}
%\begin{proof}
%See Lemma 4 in \cite{Yi2022}.
%\end{proof}
The evolution of local dual variables is then characterized.
\begin{lemma}\label{lem5}
Let Assumptions 1 hold and ${\gamma _t}{\beta _t} \le 1$, $t \in {\mathbb{N}_ + }$. For all $i \in [ n ]$ and $t \in {\mathbb{N}_ + }$, ${q_{i,t}}$ induced by Algorithm 1 satisfy\par\nobreak\vspace{-10pt}
\begin{small}
\begin{flalign}
{\Delta _{i,t}}( {{\mu _i}} ) \le 2F_1^2{\gamma _t} + \psi _{i,t - 1}^T{[ {{c_{i,t - 1}}( {{x_{i,t - 1}}} )} ]_ + } + \frac{1}{2}{\beta _t}{\| {{\mu _i}} \|^2}. \label{lemma4-eq1}
\end{flalign}
\end{small}
\end{lemma}
\begin{proof}
Based on mathematical induction, we first prove\par\nobreak\vspace{-10pt}
\begin{small}
\begin{flalign}
\| {{\beta _t}{q_{i,t}}} \| \le {F_1}. \label{lemma4-proof-eq1}
\end{flalign}
\end{small}%

It follows from ${q_{i,1}} = {\mathbf{0}_{{m_i}}}$, $\forall i \in [ n ]$ that $\| {{\beta _1}{q_{i,1}}} \| \le {F_1}$, $\forall i \in [ n ]$.
Let us assume that $\| {{\beta _t}{q_{i,t}}} \| \le {F_1}$ for $i \in [ n ]$ hold. We need to show that $\| {{\beta _{t + 1}}{q_{i,t + 1}}} \| \le {F_1}$ for $i \in [ n ]$ remain true.
%\TY{I did not quite follow the above logic. Please use the math induction.}

From \eqref{ass-eq2b}, we have\par\nobreak\vspace{-10pt}
\begin{small}
\begin{flalign}
\| {{{[ {{c_{i,t}}( {{x_{i,t}}} )} ]}_ + }} \| \le {F_1}. \label{lemma4-proof-eq2}
\end{flalign}
\end{small}%
From \eqref{ass-eq2b}, \eqref{Algorithm1-eq5}, \eqref{lemma1-eq1} and $\| {{\beta _t}{q_{i,t}}} \| \le {F_1}$, for $ i \in [ n ]$, we have\par\nobreak\vspace{-10pt}
\begin{small}
\begin{flalign}
\nonumber
\| {{q_{i,t + 1}}} \| &= \big\| {{{\big[ {( {1 - {\beta _{t + 1}}{\gamma _{t + 1}}} ){q_{i,t}} + {\gamma _{t + 1}}{{[ {{c_{i,t}}( {{x_{i,t}}} )} ]}_ + }} \big]}_ + }} \big\|\\
\nonumber
&  \le \| {( {1 - {\beta _{t + 1}}{\gamma _{t + 1}}} ){q_{i,t}} + {\gamma _{t + 1}}{F_1}} \|\\
\nonumber
&  \le ( {1 - {\gamma _{t + 1}}{\beta _{t + 1}}} )\frac{{{F_1}}}{{{\beta _t}}} + {\gamma _{t + 1}}{F_1}\\
&  \le ( {1 - {\gamma _{t + 1}}{\beta _{t + 1}}} )\frac{{{F_1}}}{{{\beta _{t + 1}}}} + {\gamma _{t + 1}}{F_1}
= \frac{{{F_1}}}{{{\beta _{t + 1}}}}. \label{lemma4-proof-eq3}
\end{flalign}
\end{small}%
Hence, \eqref{lemma4-proof-eq1} follows.

For any ${\mu _i} \in \mathbb{R}_ + ^{{m_i}}$, we have \par\nobreak\vspace{-10pt}
\begin{small}
\begin{flalign}
\nonumber
{\| \psi _{i,t} \|^2}
\nonumber
&   = {\big\| {{{\big[ {( {1 - {\beta _t}{\gamma _t}} ){q_{i,t - 1}} + {\gamma _t}{{[ {{c_{i,t - 1}}( {{x_{i,t - 1}}} )} ]}_ + }} \big]}_ + } - {{[ {{\mu _i}} ]}_ + }} \big\|^2}\\
\nonumber
&\le {\| {( {1 - {\beta _t}{\gamma _t}} ){q_{i,t - 1}} + {\gamma _t}{{[ {{c_{i,t - 1}}( {{x_{i,t - 1}}} )} ]}_ + } - {\mu _i}} \|^2}\\
\nonumber
&  = {\big\| {{q_{i,t - 1}} - {\mu _i} + {\gamma _t}\big( {{{[ {{c_{i,t - 1}}( {{x_{i,t - 1}}} )} ]}_ + } - {\beta _t}{q_{i,t - 1}}} \big)} \big\|^2}\\
\nonumber
&  = {\| \psi _{i,t - 1} \|^2} + \gamma _t^2{\| {{{[ {{c_{i,t - 1}}( {{x_{i,t - 1}}} )} ]}_ + } - {\beta _t}{q_{i,t - 1}}} \|^2}\\
\nonumber
&\;\;\;\;\;  + 2{\gamma _t}{\psi _{i,t - 1}^T}{[ {{c_{i,t - 1}}( {{x_{i,t - 1}}} )} ]_ + } - 2{\beta _t}{\gamma _t}{\psi _{i,t - 1}^T}{q_{i,t - 1}} \\
\nonumber
&  \le {\| \psi _{i,t - 1} \|^2} + 2\gamma _t^2{F_1}  + 2{\gamma _t}{\psi _{i,t - 1}^T}{[ {{c_{i,t - 1}}( {{x_{i,t - 1}}} )} ]_ + }\\
&\;\;\;\;\;  + {\beta _t}{\gamma _t}( {{{\| {{\mu _i}} \|}^2} - {{\| \psi _{i,t - 1} \|}^2}} ), \label{lemma4-proof-eq4}
\end{flalign}
\end{small}%
where the first inequality holds since the projection ${[  \cdot  ]_ + }$ is nonexpansive and \eqref{Algorithm1-eq5} holds, and the second inequality holds due to \eqref{lemma4-proof-eq1} and \eqref{lemma4-proof-eq2}. Hence, \eqref{lemma4-eq1} follows.
\end{proof}

We next analyze network regret at one iteration in the following lemma.
\begin{lemma}\label{lem7}
Let Assumptions 1 and 3 hold. For all $i \in [ n ]$ and arbitrary sequence $\{ {y_t} \in \mathbb{X}\} $, $\{ {{x_{i,t}}} \}$ induced by Algorithm~1 satisfy\par\nobreak\vspace{-10pt}
\begin{small}
\begin{flalign}
\nonumber
&\;\;\;\;\; \frac{1}{n}\sum\limits_{i = 1}^n {{l_t}( {{x_{i,t}}} )}  - {l_t}( {{y_t}} ) \\
\nonumber
& \le \frac{1}{n}\sum\limits_{i = 1}^n {\varepsilon _{i,t}^1}  + \frac{1}{{2n{\alpha _{t + 1}}}}\sum\limits_{i = 1}^n {{\mathbf{E}_{{\mathfrak{U}_t}}}[ {\varepsilon _{i,t}^2} ]} - \frac{1}{n}\sum\limits_{i = 1}^n {\frac{{{\mathbf{E}_{{\mathfrak{U}_t}}}[ {{{\| {\varepsilon _{i,t}^e} \|}^2}} ]}}{{2{\alpha _{t + 1}}}}} \\
\nonumber
&\;\;\;\;\;  + \frac{1}{n}\sum\limits_{i = 1}^n {( {2{F_2}\| {{e_{i,t}} - {{\bar e}_t}} \| + \frac{p{F_1}}{{{\delta _t}}}{\mathbf{E}_{{\mathfrak{U}_t}}}[ {\| \varepsilon _{i,t} \|} ]} )} \\
\nonumber
&\;\;\;\;\;  + \frac{1}{n}\sum\limits_{i = 1}^n {{\mathbf{E}_{{\mathfrak{U}_t}}}\big[ {\langle {\hat \partial {{[ {{c_{i,t}}( {{e_{i,t}}} )} ]}_ + }{q_{i,t}},\varepsilon _{i,t}} \rangle } \big]} \\
&\;\;\;\;\;  + \frac{1}{n}\sum\limits_{i = 1}^n {{F_2} ( {\psi _t^1 + {\delta _t}} ) \| {{q_{i,t}}} \|}  + {F_2}( {\psi _t^1 + 5{\delta _t}} ). \label{lemma5-eq1}
\end{flalign}
\end{small}%
\end{lemma}
\begin{proof}
For all $i \in [ n ]$, $t \in {\mathbb{N}_ + }$, $x,y \in \mathbb{X}$, it follows from the third item in Assumption 1 that \par\nobreak\vspace{-10pt}
\begin{small}
\begin{subequations}
\begin{flalign}
| {{l_{i,t}}( x ) - {l_{i,t}}( y )} | &\le {F_2}\| {x - y} \|, \label{lemma5-proof-eq1a}\\
\| {{c_{i,t}}( x ) - {c_{i,t}}( y )} \| &\le {F_2}\| {x - y} \|. \label{lemma5-proof-eq1b}
\end{flalign}
\end{subequations}
\end{small}%

From ${q_{i,t}} \ge {\mathbf{0}_{{m_i}}}$, ${e_{i,t}},{{\hat y}_t} \in ( {1 - {\xi _t}} )\mathbb{X} \subseteq \mathbb{X}$, \eqref{ass-eq1}, \eqref{lemma2-eq1b}, \eqref{lemma2-eq1e}, \eqref{lemma5-proof-eq1a} and \eqref{lemma5-proof-eq1b}, we have\par\nobreak\vspace{-10pt}
\begin{small}
\begin{flalign}
&\;\;\;\;\; | {{l_{i,t}}( {{x_{i,t}}} ) - {l_{i,t}}( {{e_{i,t}}} )} | \le {F_2}\| {{x_{i,t}} - {e_{i,t}}} \| \le {F_2}{\delta _t}, \label{lemma5-proof-eq2}\\
&\;\;\;\;\; {l_{i,t}}( {{e_{i,t}}} ) \le {{\hat l}_{i,t}}( {{e_{i,t}}} ), \label{lemma5-proof-eq3}\\
&\;\;\;\;\; {{\hat l}_{i,t}}( {{{\hat y}_t}} ) - {l_{i,t}}( {{y_t}} )\le F_2\psi _t^1, \label{lemma5-proof-eq4}\\
&\;\;\;\;\; q_{i,t}^T{[ {{c_{i,t}}( {{e_{i,t}}} )} ]_ + } \le q_{i,t}^T{[ {{{\hat c}_{i,t}}( {{e_{i,t}}} )} ]_ + }, \label{lemma5-proof-eq5}\\
&\;\;\;\;\; \| {{c_{i,t}}( {{x_{i,t}}} ) - {c_{i,t}}( {{e_{i,t}}} )} \| \le {F_2}\| {{x_{i,t}} - {e_{i,t}}} \| \le {F_2}{\delta _t}, \label{lemma5-proof-eq6}\\
&\;\;\;\;\; q_{i,t}^T{[ {{{\hat c}_{i,t}}( {{{\hat y}_t}} )} ]_ + } \le q_{i,t}^T{[ {{c_{i,t}}( {{y_t}} )} ]_ + } + {F_2}\psi _t^1\| {{q_{i,t}}} \|. \label{lemma5-proof-eq7}
%\nonumber
%\le q_{i,t}^T{[ {{c_{i,t}}( {{{\hat y}_t}} )} ]_ + } + {F_2}{\delta _t}\| {{q_{i,t}}} \|\\
%\nonumber
%&= q_{i,t}^T\big( {{{[ {{c_{i,t}}( {{y_t}} )} ]}_ + } + {{[ {{c_{i,t}}( {{{\hat y}_t}} )} ]}_ + } - {{[ {{c_{i,t}}( {{y_t}} )} ]}_ + }} \big) + {F_2}{\delta _t}\| {{q_{i,t}}} \| \\
%\nonumber
%&\le q_{i,t}^T{[ {{c_{i,t}}( {{y_t}} )} ]_ + } + {F_2}\| {{{\hat y}_t} - {y_t}} \|\| {{q_{i,t}}} \| + {F_2}{\delta _t}\| {{q_{i,t}}} \|\\
%&
\end{flalign}
\end{small}%
In particular, we have \eqref{lemma5-proof-eq4} since
\begin{small}
\begin{flalign}
\nonumber
&\;\;\;\;\;{{\hat l}_{i,t}}( {{{\hat y}_t}} ) - {l_{i,t}}( {{y_t}} )\\
\nonumber
& = {l_{i,t}}( {{{\hat y}_t}} ) - {l_{i,t}}( {{y_t}} ) + {{\hat l}_{i,t}}( {{{\hat y}_t}} ) - {l_{i,t}}( {{{\hat y}_t}} )\\
\nonumber
& \le {F_2}\| {{{\hat y}_t} - {y_t}} \| + {{\hat l}_{i,t}}( {{{\hat y}_t}} ) - {l_{i,t}}( {{{\hat y}_t}} )\\
\nonumber
& \le {F_2}R( \mathbb{X} ){\xi _t} + {F_2}{\delta _t}  = {F_2}\psi _t^1,
\end{flalign}
\end{small}%
where the first inequality holds due to \eqref{lemma5-proof-eq1a}, and the last inequality holds due to \eqref{ass-eq1} and \eqref{lemma2-eq1b}.

We have \eqref{lemma5-proof-eq7} since
\begin{small}
\begin{flalign}
\nonumber
&\;\;\;\;\;q_{i,t}^T{[ {{{\hat c}_{i,t}}( {{{\hat y}_t}} )} ]_ + } \\
\nonumber
& \le q_{i,t}^T{[ {{c_{i,t}}( {{{\hat y}_t}} )} ]_ + } + {F_2}{\delta _t}\| {{q_{i,t}}} \|\\
\nonumber
& = q_{i,t}^T\big( {{{[ {{c_{i,t}}( {{y_t}} )} ]}_ + } + {{[ {{c_{i,t}}( {{{\hat y}_t}} )} ]}_ + } - {{[ {{c_{i,t}}( {{y_t}} )} ]}_ + }} \big) + {F_2}{\delta _t}\| {{q_{i,t}}} \|\\
\nonumber
&  \le q_{i,t}^T{[ {{c_{i,t}}( {{y_t}} )} ]_ + } + q_{i,t}^T\big( {{c_{i,t}}( {{{\hat y}_t}} ) - {c_{i,t}}( {{y_t}} )} \big) + {F_2}{\delta _t}\| {{q_{i,t}}} \| \\
\nonumber
& \le q_{i,t}^T{[ {{c_{i,t}}( {{y_t}} )} ]_ + } + {F_2}\| {{{\hat y}_t} - {y_t}} \|\| {{q_{i,t}}} \| + {F_2}{\delta _t}\| {{q_{i,t}}} \|\\
\nonumber
& \le q_{i,t}^T{[ {{c_{i,t}}( {{y_t}} )} ]_ + } + {F_2}( {R( \mathbb{X} ){\xi _t} + {\delta _t}} )\| {{q_{i,t}}} \|  \\
\nonumber
& = q_{i,t}^T{[ {{c_{i,t}}( {{y_t}} )} ]_ + } + {F_2}\psi _t^1\| {{q_{i,t}}} \|,
\end{flalign}
\end{small}%
where the first inequality holds due to \eqref{lemma2-eq1e}, the second inequality holds since the projection ${\left[  \cdot  \right]_ + }$ is nonexpansive, the third inequality holds due to \eqref{lemma5-proof-eq1b}, and the last inequality holds due to \eqref{ass-eq1}.

We have\par\nobreak\vspace{-10pt}
\begin{small}
\begin{flalign}
\nonumber
&\;\;\;\;\; \frac{1}{n}\sum\limits_{i = 1}^n {{l_t}( {{x_{i,t}}} )}  = \frac{1}{n}\sum\limits_{i = 1}^n {\Big( {\frac{1}{n}\sum\limits_{j = 1}^n {{l_{j,t}}( {{x_{i,t}}} )} } \Big)}\\
\nonumber
& = \frac{1}{n}\sum\limits_{i = 1}^n {{l_{i,t}}( {{x_{i,t}}} )}  + \frac{1}{{{n^2}}}\sum\limits_{i = 1}^n {\sum\limits_{j = 1}^n {\big( {{l_{j,t}}( {{x_{i,t}}} ) - {l_{j,t}}( {{x_{j,t}}} )} \big)} } \\
\nonumber
& \le \frac{1}{n}\sum\limits_{i = 1}^n {{l_{i,t}}( {{x_{i,t}}} )}  + \frac{1}{{{n^2}}}\sum\limits_{i = 1}^n {\sum\limits_{j = 1}^n {{F_2}\| {{x_{i,t}} - {x_{j,t}}} \|} } \\
\nonumber
& = \frac{1}{n}\sum\limits_{i = 1}^n {{l_{i,t}}( {{x_{i,t}}} )} \\
\nonumber
&\;\;\;\;\;
 + \frac{1}{{{n^2}}}\sum\limits_{i = 1}^n {\sum\limits_{j = 1}^n {{F_2}\| {{x_{i,t}} - {x_{j,t}} - \frac{1}{n}\sum\limits_{s = 1}^n {{x_{s,t}}} + \frac{1}{n}\sum\limits_{s = 1}^n {{x_{s,t}}}} \|} } \\
\nonumber
& \le \frac{1}{n}\sum\limits_{i = 1}^n {{l_{i,t}}( {{x_{i,t}}} )}  + \frac{{2{F_2}}}{n}\sum\limits_{i = 1}^n {\| {{x_{i,t}} - \frac{1}{n}\sum\limits_{j = 1}^n {{x_{j,t}}} } \|} \\
\nonumber
& \le \frac{1}{n}\sum\limits_{i = 1}^n {{{\hat l}_{i,t}}( {{e_{i,t}}} )}  \\
\nonumber
&\;\;\;\;\; + \frac{{2{F_2}}}{n}\sum\limits_{i = 1}^n {\| {{e_{i,t}} - {{\bar e}_t} + \frac{{{\delta _t}}}{n}\sum\limits_{j = 1}^n {( {{u_{i,t}} - {u_{j,t}}} )} } \|}  + {F_2}{\delta _t} \\
& \le \frac{1}{n}\sum\limits_{i = 1}^n {{{\hat l}_{i,t}}( {{e_{i,t}}} )}  + \frac{{2{F_2}}}{n}\sum\limits_{i = 1}^n {\| {{e_{i,t}} - {{\bar e}_t}} \|}  + 5{F_2}{\delta _t}, \label{lemma5-proof-eq8}
\end{flalign}
\end{small}%
where the first inequality holds due to \eqref{lemma5-proof-eq1a}, the third inequality holds due to \eqref{lemma2-eq1b}, ${e_{i,t}} \in ( {1 - {\xi _t}} )\mathbb{X}$, \eqref{Algorithm1-eq4}, and \eqref{lemma5-proof-eq2}, and the last inequality holds due to $\| {{u_{i,t}}} \| = 1$.

From ${e_{i,t}}$, ${\hat y_t} \in ( {1 - {\xi _t}} )\mathbb{X}$, we have\par\nobreak\vspace{-10pt}
\begin{small}
\begin{flalign}
\nonumber
&\;\;\;\;\; {{\hat l}_{i,t}}( {{e_{i,t}}} ) - {{\hat l}_{i,t}}( {{{\hat y}_t}} ) \\
\nonumber
&\le \langle {\partial {{\hat l}_{i,t}}( {{e_{i,t}}} ),{e_{i,t}} - {{\hat y}_t}} \rangle \\
\nonumber
&= \langle {{{\rm{\mathbf{E}}}_{{\mathfrak{U}_t}}}[ {\hat \partial {l_{i,t}}( {{e_{i,t}}} )} ],{e_{i,t}} - {{\hat y}_t}} \rangle \\
\nonumber
& = {{\rm{\mathbf{E}}}_{{\mathfrak{U}_t}}}[ {\langle {\hat \partial {l_{i,t}}( {{e_{i,t}}} ),{e_{i,t}} - {{\hat y}_t}} \rangle } ]\\
\nonumber
&  = {{\rm{\mathbf{E}}}_{{\mathfrak{U}_t}}}[ {\langle {\hat \partial {l_{i,t}}( {{e_{i,t}}} ),{\varepsilon _{i,t}}} \rangle } + \langle {\hat \partial {l_{i,t}}( {{e_{i,t}}} ),{e_{i,t + 1}} - {{\hat y}_t}} \rangle  ]\\
& \le {{\rm{\mathbf{E}}}_{{\mathfrak{U}_t}}}[ {\frac{{p{F_1}}}{{{\delta _t}}}\| \varepsilon _{i,t} \| + \langle {\hat \partial {l_{i,t}}( {{e_{i,t}}} ),{e_{i,t + 1}} - {{\hat y}_t}} \rangle } ], \label{lemma5-proof-eq9}
\end{flalign}
\end{small}%
where the first inequality holds since the function ${\hat l_{i,t}}( x )$ is convex for any $x \in ( {1 - {\xi _t}} )\mathbb{X}$, the first equality holds due to \eqref{lemma2-eq1a}, the second equality holds since $e_{i,t}$ and ${{\hat y}_t}$ are independent of ${\mathfrak{U}_t}$, and the last inequality holds due to \eqref{lemma2-eq1c}.

From \eqref{Algorithm1-eq2}, we have\par\nobreak\vspace{-10pt}
\begin{small}
\begin{flalign}
\nonumber
&\;\;\;\;\; \langle {\hat \partial {l_{i,t}}( {{e_{i,t}}} ),{e_{i,t + 1}} - {{\hat y}_t}} \rangle \\
\nonumber
& = \langle {{{\hat \omega }_{i,t + 1}},{e_{i,t + 1}} - {{\hat y}_t}} \rangle  + \langle {\hat \partial {{[ {{c_{i,t}}( {{e_{i,t}}} )} ]}_ + }{q_{i,t}},{{\hat y}_t} - {e_{i,t + 1}}} \rangle \\
\nonumber
& = \langle {\hat \partial {{[ {{c_{i,t}}( {{e_{i,t}}} )} ]}_ + }{q_{i,t}},{{\hat y}_t} - {e_{i,t}}} \rangle \\
&\;\;\;\;\; + \langle {\hat \partial {{[ {{c_{i,t}}( {{e_{i,t}}} )} ]}_ + }{q_{i,t}},\varepsilon _{i,t}} \rangle + \langle {{{\hat \omega }_{i,t + 1}},{e_{i,t + 1}} - {{\hat y}_t}} \rangle. \label{lemma5-proof-eq10}
\end{flalign}
\end{small}%
From ${e_{i,t}}$, ${\hat y_t} \in ( {1 - {\xi _t}} )\mathbb{X}$, we have\par\nobreak\vspace{-10pt}
\begin{small}
\begin{flalign}
\nonumber
&\;\;\;\;\; {{\rm{\mathbf{E}}}_{{\mathfrak{U}_t}}}\big[ {\langle {\hat \partial {{[ {{c_{i,t}}( {{e_{i,t}}} )} ]}_ + }{q_{i,t}},{{\hat y}_t} - {e_{i,t}}} \rangle } \big] \\
\nonumber
& = \big\langle {{{\rm{\mathbf{E}}}_{{\mathfrak{U}_t}}}\big[ {\hat \partial {{[ {{c_{i,t}}( {{e_{i,t}}} )} ]}_ + }} \big]{q_{i,t}},{{\hat y}_t} - {e_{i,t}}} \big\rangle \\
\nonumber
&= \langle {\partial {{[ {{{\hat c}_{i,t}}( {{e_{i,t}}} )} ]}_ + }{q_{i,t}},{{\hat y}_t} - {e_{i,t}}} \rangle \\
\nonumber
& \le q_{i,t}^T{[ {{{\hat c}_{i,t}}( {{{\hat y}_t}} )} ]_ + } - q_{i,t}^T{[ {{{\hat c}_{i,t}}( {{e_{i,t}}} )} ]_ + }\\
\nonumber
& \le q_{i,t}^T{[ {{c_{i,t}}( {{y_t}} )} ]_ + } - q_{i,t}^T{[ {{c_{i,t}}( {{e_{i,t}}} )} ]_ + } + {F_2}\psi _t^1\| {{q_{i,t}}} \| 
\end{flalign}
\end{small}%
\begin{small}
\begin{flalign}
\nonumber
& = q_{i,t}^T\big( {{{[ {{c_{i,t}}( {{y_t}} )} ]}_ + } - {{[ {{c_{i,t}}( {{e_{i,t}}} )} ]}_ + } + {{[ {{c_{i,t}}( {{x_{i,t}}} )} ]}_ + } - {{[ {{c_{i,t}}( {{x_{i,t}}} )} ]}_ + }} \big)\\
\nonumber
&\;\;\;\;\; + {F_2}\psi _t^1\| {{q_{i,t}}} \| \\
\nonumber
& \le \varepsilon _{i,t}^1 + q_{i,t}^T\big( {{c_{i,t}}( {{x_{i,t}}} ) - {c_{i,t}}( {{e_{i,t}}} )} \big) + {F_2}\psi _t^1\| {{q_{i,t}}} \|\\
& \le \varepsilon _{i,t}^1 + {F_2}( {\psi _t^1 + {\delta _t}} )\| {{q_{i,t}}} \|. \label{lemma5-proof-eq11}
\end{flalign}
\end{small}%
where the first equality holds since ${e_{i,t}}$, ${q_{i,t}}$, and ${\hat y_t}$ are independent of ${\mathfrak{U}_t}$, the second equality due to \eqref{lemma2-eq1d}, the first inequality holds since ${q_{i,t}} \ge {\mathbf{0}_{{m_i}}}$ and ${[ {{{\hat c}_{i,t}}( x )} ]_ + }$ is convex on $( {1 - {\xi _t}} )\mathbb{X}$, the second inequality holds due to \eqref{lemma5-proof-eq5} and \eqref{lemma5-proof-eq7}, the third inequality holds since the projection ${\left[  \cdot  \right]_ + }$ is nonexpansive, and the last inequality holds due to \eqref{lemma5-proof-eq6}.

Applying \eqref{lemma1-eq2} to the update \eqref{Algorithm1-eq3}, we get\par\nobreak\vspace{-10pt}
\begin{small}
\begin{flalign}
\nonumber
&\;\;\;\;\; \langle {{{\hat \omega }_{i,t + 1}},{e_{i,t + 1}} - {{\hat y}_t}} \rangle \\
\nonumber
& \le \frac{1}{{2{\alpha _{t + 1}}}}( {{{\| {{{\hat y}_t} - {z_{i,t + 1}}} \|}^2} - {{\| {{{\hat y}_t} - {e_{i,t + 1}}} \|}^2} - {{\| {\varepsilon _{i,t}^e} \|}^2}} ) \\
\nonumber
&  = \frac{1}{{2{\alpha _{t + 1}}}}( {{{\| {{{\hat y}_t} - {z_{i,t + 1}}} \|}^2} - {{\| {{{\hat y}_{t + 1}} - {z_{i,t + 2}}} \|}^2} - } {\| {\varepsilon _{i,t}^e} \|^2} \\
\nonumber
&\;\;\;\;\;  + {{\| {{{\hat y}_{t + 1}} - {z_{i,t + 2}}} \|}^2} - {{\| {{{\hat y}_t} - {e_{i,t + 1}}} \|}^2} ) \\
\nonumber
&  = \frac{1}{{2{\alpha _{t + 1}}}}( {{{\| {{{\hat y}_t} - {z_{i,t + 1}}} \|}^2} - {{\| {{{\hat y}_{t + 1}} - {z_{i,t + 2}}} \|}^2} - } {\| {\varepsilon _{i,t}^e} \|^2} \\
\nonumber
&\;\;\;\;\;  + {{\| {{{\hat y}_{t + 1}} - \sum\limits_{j = 1}^n {{{[ {{W_{t + 1}}} ]}_{ij}}{e_{j,t + 1}}} } \|^2}} - {{\| {{{\hat y}_t} - {e_{i,t + 1}}} \|}^2} ) \\
\nonumber
& \le \frac{1}{{2{\alpha _{t + 1}}}}( {{\| {{{\hat y}_t} - {z_{i,t + 1}}} \|}^2} - {{\| {{{\hat y}_{t + 1}} - {z_{i,t + 2}}} \|}^2} -  \| {\varepsilon _{i,t}^e} \|^2 \\
&\;\;\;\;\;  + \sum\limits_{j = 1}^n {{{[ {{W_{t + 1}}} ]}_{ij}}{{\| {{{\hat y}_{t + 1}} - {e_{j,t + 1}}} \|}^2}}  - {{\| {{{\hat y}_t} - {e_{i,t + 1}}} \|}^2} ), \label{lemma5-proof-eq12}
\end{flalign}
\end{small}%
where the second equality due to \eqref{Algorithm1-eq1}, and the last inequality holds since ${\left\|  \cdot  \right\|^2}$ is convex and ${W_{t + 1}}$ is doubly stochastic. 

From \eqref{lemma5-proof-eq4}, \eqref{lemma5-proof-eq8}--\eqref{lemma5-proof-eq12}, and noting that ${W_{t + 1}}$ is doubly stochastic, we have \eqref{lemma5-eq1}.
\end{proof}

Finally, we are ready to prove Theorem 1.

\begin{proof}
($\mathbf{i}$) Let us first establish the dynamic network regret bound.

It follows from \eqref{ass-eq1} and ${\hat y_t} = ( {1 - {\xi _t}} ){y_t}$ that\par\nobreak\vspace{-10pt}
\begin{small}
\begin{flalign}
\nonumber
&\;\;\;\;\; {\| {{{\hat y}_{t + 1}} - {e_{j,t + 1}}} \|^2} - {\| {{{\hat y}_t} - {e_{j,t + 1}}} \|^2} \\
\nonumber
& \le \| {{{\hat y}_{t + 1}} - {{\hat y}_t}} \|\| {{{\hat y}_{t + 1}} + {{\hat y}_t} - 2{e_{j,t + 1}}} \| \\
\nonumber
&  = 4R( \mathbb{X} )\| {( {1 - {\xi _{t + 1}}} )( {{y_{t + 1}} - {y_t}} ) + ( {{\xi _t} - {\xi _{t + 1}}} ){y_t}} \| \\
& \le 4R( \mathbb{X} )\| {{y_{t + 1}} - {y_t}} \| + 4R{( \mathbb{X} )^2}( {{\xi _t} - {\xi _{t + 1}}} ). \label{theorem1-proof-eq1}
\end{flalign}
\end{small}%
From \eqref{lemma2-eq1f}, we have\par\nobreak\vspace{-10pt}
\begin{small}
\begin{flalign}
\nonumber
\langle {\hat \partial {{[ {{c_{i,t}}( {{e_{i,t}}} )} ]}_ + }{q_{i,t}}, \varepsilon _{i,t}} \rangle
& \le \big\| {\hat \partial {{[ {{g_{i,t}}( {{e_{i,t}}} )} ]}_ + }} \big\|\big\| {{q_{i,t}}} \big\|\big\| \varepsilon _{i,t} \big\| \\
& \le \frac{{p{F_1}}}{{{\delta _t}}}\| {{q_{i,t}}} \| \| \varepsilon _{i,t} \|. \label{theorem1-proof-eq2}
\end{flalign}
\end{small}%
From \eqref{lemma4-eq1}, \eqref{lemma5-eq1}, \eqref{theorem1-proof-eq1}, \eqref{theorem1-proof-eq2}, and ${c_{i,t}}( {{y_t}} ) \le {\mathbf{0}_{{m_i}}}$, we have\par\nobreak\vspace{-10pt}
\begin{small}
\begin{flalign}
\nonumber
&\;\;\;\;\; \frac{1}{n}\sum\limits_{i = 1}^n {\big( {{\Delta _{i,t + 1}}( {{\mu _i}} ) + \mu _i^T{{[ {{c_{i,t}}( {{x_{i,t}}} )} ]}_ + } - \frac{1}{2}{\beta _{t + 1}}{{\| {{\mu _i}} \|}^2}} \big)} \\
\nonumber
&\;\;\;\;\; + \frac{1}{n}\sum\limits_{i = 1}^n {{l_t}( {{x_{i,t}}} )}  - {l_t}( {{y_t}} )\\
\nonumber
& \le 2F_1^2{\gamma _{t + 1}} + \frac{1}{n}\sum\limits_{i = 1}^n {{{\tilde \Delta }_{i,t + 1}}} 
\end{flalign}
\end{small}%
\begin{small}
\begin{flalign}
\nonumber
& + \frac{{2R( \mathbb{X} )}}{{{\alpha _{t + 1}}}}{\mathbf{E}_{{\mathfrak{U}_t}}}[ {\| {{y_{t + 1}} - {y_t}} \| + R( \mathbb{X} )( {{\xi _t} - {\xi _{t + 1}}} )} ]\\
& + \frac{1}{{2n{\alpha _{t + 1}}}}\sum\limits_{i = 1}^n {{\mathbf{E}_{{\mathfrak{U}_t}}}[ {{{\| {{{\hat y}_t} - {z_{i,t+1}}} \|}^2} - {{\| {{{\hat y}_{t + 1}} - {z_{i,t + 2}}} \|}^2}} ]}, \label{theorem1-proof-eq3}
\end{flalign}
\end{small}%
where\par\nobreak\vspace{-10pt}
\begin{small}
\begin{flalign}
\nonumber
&\;\;\;\;\;{{\tilde \Delta }_{i,t + 1}} \\
\nonumber
&= 2{F_2}\| {{e_{i,t}} - {{\bar e}_t}} \| + \frac{p{F_1}}{{{\delta _t}}}{\mathbf{E}_{{\mathfrak{U}_t}}}[ {\| \varepsilon _{i,t} \|} ] + \frac{p{F_1}}{{{\delta _t}}}{\mathbf{E}_{{\mathfrak{U}_t}}}[ {\| {{q_{i,t}}} \|\| \varepsilon _{i,t} \|} ]  \\
\nonumber
&\;\;\;\;\; + {F_2} ({\psi _t^1 + {\delta _t}}) \| {{q_{i,t}}} \|
  + {F_2}( \psi _t^1 + 5{\delta _t} ) - \frac{{{\mathbf{E}_{{\mathfrak{U}_t}}}[ {{{\| {\varepsilon _{i,t}^e} \|}^2}} ]}}{{2{\alpha _{t + 1}}}}.
\end{flalign}
\end{small}%

From $\{ {{\alpha _t}} \}$ is non-increasing and \eqref{ass-eq1}, we have\par\nobreak\vspace{-10pt}
\begin{small}
\begin{flalign}
\nonumber
&\;\;\;\;\; \sum\limits_{t = 1}^T {\frac{1}{{{\alpha _{t + 1}}}}( {{{\| {{{\hat y}_t} - {z_{i,t + 1}}} \|}^2} - {{\| {{{\hat y}_{t + 1}} - {z_{i,t + 2}}} \|}^2}} )} \\
\nonumber
&  = \sum\limits_{t = 1}^T {\big( {\frac{1}{{{\alpha _t}}}{{\| {{{\hat y}_t} - {z_{i,t + 1}}} \|}^2} - \frac{1}{{{\alpha _{t + 1}}}}{{\| {{{\hat y}_{t + 1}} - {z_{i,t + 2}}} \|}^2}} \big.} \\
\nonumber
&\;\;\;\;\;  + \big. {( {\frac{1}{{{\alpha _{t + 1}}}} - \frac{1}{{{\alpha _t}}}} ){{\| {{{\hat y}_t} - {z_{i,t + 1}}} \|}^2}} \big) \\
\nonumber
& \le \frac{1}{{{\alpha _1}}}{\left\| {{{\hat y}_1} - {z_{i,2}}} \right\|^2} - \frac{1}{{{\alpha _{T + 1}}}}{\left\| {{{\hat y}_{T + 1}} - {z_{i,T + 2}}} \right\|^2}\\
&  + \big( {\frac{1}{{{\alpha _{T + 1}}}} - \frac{1}{{{\alpha _1}}}} \big)4R{\left( \mathbb{X} \right)^2} 
\le \frac{{4R{{( \mathbb{X} )}^2}}}{{{\alpha _{T + 1}}}}. \label{theorem1-proof-eq4}
\end{flalign}
\end{small}%

Using \eqref{theorem1-proof-eq4}, selecting ${\mu _i} = {\mathbf{0}_{{m_i}}}$, and letting ${y_{T + 1}} = {y_T}$, summing \eqref{theorem1-proof-eq3} over $t \in [ T ]$, we have\par\nobreak\vspace{-10pt}
\begin{small}
\begin{flalign}
\nonumber
&\;\;\;\;\; \mathbf{E}[ {\rm{Net}\mbox{-}\rm{Reg}}( {\{ {{x_{i,t}}} \},{y_{[ T ]}}} ) + \frac{1}{n}\sum\limits_{t = 1}^T {\sum\limits_{i = 1}^n {{\Delta _{i,t + 1}}( {{\mathbf{0}_{{m_i}}}} )} }] \\
\nonumber
& \le \mathbf{E}[  2F_1^2\sum\limits_{t = 1}^T {{\gamma _{t + 1}}}   + \frac{1}{n}\sum\limits_{t = 1}^T {\sum\limits_{i = 1}^n {{{\tilde \Delta }_{i,t + 1}}} } \\
&\;\;\;\;\; + \frac{{2R{{( \mathbb{X} )}^2}}}{{{\alpha _{T + 1}}}} + \frac{{2R( \mathbb{X} )}}{{{\alpha _T}}}{P_T}
+  \sum\limits_{t = 1}^T {\frac{{2R{{( \mathbb{X} )}^2}( {{\xi _t} - {\xi _{t + 1}}} )}}{{{\alpha _{t + 1}}}}} ]. \label{theorem1-proof-eq5}
\end{flalign}
\end{small}%

Next, the lower bound for $\sum\nolimits_{t = 1}^T {\sum\nolimits_{i = 1}^n {\mathbf{E}[ {{\Delta _{i,t + 1}}( {{\mathbf{0}_{{m_i}}}} )} ]} } $ and the upper bound for $\sum\nolimits_{t = 1}^T {\sum\nolimits_{i = 1}^n {\mathbf{E}[ {{{\tilde \Delta }_{i,t + 1}}} ]} }$ are established to get \eqref{theorem1-eq2}, respectively.

($\mathbf{i}\mbox{-}\mathbf{1}$)
For any $T \in {\mathbb{N}_ + }$,  we have\par\nobreak\vspace{-10pt}
\begin{small}
\begin{flalign}
\sum\limits_{t = 1}^T {{\Delta _{i,t + 1}}( {{\mu _i}} )}
\nonumber
&= \frac{1}{2}\sum\limits_{t = 1}^T {\big( {\frac{{{{\| {\psi _{i,t+1}} \|}^2}}}{{{\gamma _{t + 1}}}} - \frac{{{{\| {\psi _{i,t}} \|}^2}}}{{{\gamma _t}}}} \big.} \\
\nonumber
&\;\;\;\;\; \big. { + \big( {\frac{1}{{{\gamma _t}}} - \frac{1}{{{\gamma _{t + 1}}}} + {\beta _{t + 1}}} \big){{\| {\psi _{i,t}} \|}^2}} \big)\\
&= \frac{{{{\| \psi _{i,T+1} \|}^2}}}{{2{\gamma _{T + 1}}}} - \frac{{{{\| {{\mu _i}} \|}^2}}}{{2{\gamma _1}}} + \sum\limits_{t = 1}^T {\frac{{\psi _t^2{{\| {{\psi _{i,t}}} \|}^2}}}{2}}. \label{theorem1-proof-eq6}
\end{flalign}
\end{small}%

Let ${\mu _i} = {\mathbf{0}_{{m_i}}}$. From \eqref{theorem1-proof-eq6}, we have\par\nobreak\vspace{-10pt}
\begin{small}
\begin{flalign}
\sum\limits_{t = 1}^T {\sum\limits_{i = 1}^n {\mathbf{E}[ {{\Delta _{i,t + 1}}( {{{\mathbf{0}}_{{m_i}}}} )} ]} }
\ge \frac{1}{{2}}\sum\limits_{t = 1}^T {\sum\limits_{i = 1}^n {\psi _t^2\mathbf{E}[ {{{\| {{q_{i,t}}} \|}^2}} ]} }. \label{theorem1-proof-eq7}
\end{flalign}
\end{small}%

($\mathbf{i}\mbox{-}\mathbf{2}$)
Let $\{ {{\kappa _t},t \in [ T ]} \}$ be user-defined sequence. We have\par\nobreak\vspace{-10pt}
\begin{small}
\begin{flalign}
\hspace{-8pt}
%\sum\limits_{t = 1}^T {\sum\limits_{s = 1}^{t - 2} {{\lambda ^{t - s - 2}}\sum\limits_{j = 1}^n {\left\| {\varepsilon _{j,s}^e} \right\|} } } {\rm{ }} &= \sum\limits_{t = 1}^{T - 2} {\sum\limits_{j = 1}^n {\left\| {\varepsilon _{j,t}^e} \right\|\sum\limits_{s = 0}^{T - t - 2} {{\lambda ^s}} } } \\
%& \le \frac{1}{{1 - \lambda }}\sum\limits_{t = 1}^{T - 2} {\sum\limits_{j = 1}^n {\left\| {\varepsilon _{j,t}^e} \right\|} }, \\
%\sum\limits_{t = 1}^T {\sum\limits_{s = 1}^{t - 2} {\frac{{{\lambda ^{t - s - 2}}}}{{{\delta _t}}}\sum\limits_{j = 1}^n {\left\| {\varepsilon _{j,s}^e} \right\|} } }  &= \sum\limits_{t = 1}^{T - 2} {\sum\limits_{j = 1}^n {\left\| {\varepsilon _{j,t}^e} \right\|\sum\limits_{s = 0}^{T - t - 2} {\frac{{{\lambda ^s}}}{{{\delta _{s + t + 2}}}}} } },\\
%\sum\limits_{t = 1}^T {\sum\limits_{s = 1}^{t - 2} {\frac{{{\lambda ^{t - s - 2}}}}{{{\delta _t}{\beta _t}}}\sum\limits_{j = 1}^n {\left\| {\varepsilon _{j,s}^e} \right\|} } }  &= \sum\limits_{t = 1}^{T - 2} {\sum\limits_{j = 1}^n {\left\| {\varepsilon _{j,t}^e} \right\|\sum\limits_{s = 0}^{T - t - 2} {\frac{{{\lambda ^s}}}{{{\delta _{s + t + 2}}{\beta _{s + t + 2}}}}} } }.
\sum\limits_{t = 1}^T {\sum\limits_{k = 1}^{t - 2} {\frac{{{\lambda ^{t - k - 2}}}}{{{\kappa _t}}}\sum\limits_{j = 1}^n {\| {\varepsilon _{j,k}^e} \|} } }  = \sum\limits_{t = 1}^{T - 2} {\sum\limits_{j = 1}^n {\| {\varepsilon _{j,t}^e} \|\sum\limits_{k = 0}^{T - t - 2} {\frac{{{\lambda ^k}}}{{{\kappa _{k + t + 2}}}}} } }. \label{theorem1-proof-eq8}
\end{flalign}
\end{small}%
From \eqref{lemma3-eq1}, \eqref{theorem1-proof-eq8} and $a > 0$, setting ${\kappa _t} = 1$ for $t \in [ T ]$ gives\par\nobreak\vspace{-10pt}
\begin{small}
\begin{flalign}
\nonumber
&\;\;\;\;\; \sum\limits_{t = 1}^T {\sum\limits_{i = 1}^n {2\| {{e_{i,t}} - {{\bar e}_t}} \|} } \\
\nonumber
& \le \sum\limits_{t = 1}^T {\sum\limits_{i = 1}^n {2{\varpi _1} {\lambda ^{t - 2}} } }  + \frac{1}{n}\sum\limits_{t = 2}^T {\sum\limits_{i = 1}^n {\sum\limits_{j = 1}^n {2\| {\varepsilon _{j,t - 1}^e} \|} } } \\
\nonumber
&  + \sum\limits_{t = 2}^T {\sum\limits_{i = 1}^n {2\left\| {\varepsilon _{i,t - 1}^e} \right\|} }  + \sum\limits_{t = 1}^{T - 2} {\sum\limits_{i = 1}^n {\sum\limits_{j = 1}^n {\frac{{2\tau }}{{1 - \lambda }}\left\| {\varepsilon _{j,t}^e} \right\|} } } \\
\nonumber
& \le 2n{\varpi _2} + \sum\limits_{t = 2}^T {\sum\limits_{i = 1}^n {4\| {\varepsilon _{i,t - 1}^e} \|} }  + \sum\limits_{t = 1}^{T - 2} {\sum\limits_{i = 1}^n {\frac{{2n\tau }}{{1 - \lambda }}\| {\varepsilon _{i,t}^e} \|} } \\
\nonumber
& \le 2n{\varpi _2}  + \sum\limits_{t = 2}^T {\sum\limits_{i = 1}^n {4\| {\varepsilon _{i,t - 1}^e} \|} }  + \sum\limits_{t = 2}^T {\sum\limits_{i = 1}^n {\frac{{2n\tau }}{{1 - \lambda }}\| {\varepsilon _{i,t - 1}^e} \|} } \\
\nonumber
& \le 2n{\varpi _2}  + \sum\limits_{t = 2}^T {\sum\limits_{i = 1}^n {( {4a{F_2}{\alpha _t} + \frac{1}{{a{F_2}{\alpha _t}}}{{\| {\varepsilon _{i,t - 1}^e} \|}^2}} )} } \\
\nonumber
&\;\;\;\;\;  + \sum\limits_{t = 2}^T {\sum\limits_{i = 1}^n {( {\frac{{a{n^2}{\tau ^2}{F_2}{\alpha _t}}}{{{{( {1 - \lambda } )}^2}}} + \frac{1}{{a{F_2}{\alpha _t}}}{{\| {\varepsilon _{i,t - 1}^e} \|}^2}} )} } \\
& \le 2n{\varpi _2} + \sum\limits_{t = 2}^T {\sum\limits_{i = 1}^n {( {2a{\varpi _3}{\alpha _t} + \frac{2}{{a{F_2}{\alpha _t}}}{{\| {\varepsilon _{i,t - 1}^e} \|}^2}} )} }.  \label{theorem1-proof-eq9}
\end{flalign}
\end{small}%
From \eqref{lemma3-eq1}, \eqref{theorem1-proof-eq8} and $a > 0$, setting ${\kappa _t} = {\delta _t}$ for $t \in [ T ]$ yields\par\nobreak\vspace{-10pt}
\begin{small}
\begin{flalign}
\nonumber
&\;\;\;\;\; \sum\limits_{t = 1}^T {\sum\limits_{i = 1}^n {\frac{p}{{{\delta _t}}}\| {{e_{i,t}} - {{\bar e}_t}} \|} } \\
\nonumber
& \le \sum\limits_{t = 1}^T {\sum\limits_{i = 1}^n {\frac{{p{\varpi _1} {\lambda ^{t - 2}}}}{{{\delta _t}}} } }  + \frac{1}{n}\sum\limits_{t = 2}^T {\sum\limits_{i = 1}^n {\sum\limits_{j = 1}^n {\frac{p}{{{\delta _t}}}\| {\varepsilon _{j,t - 1}^e} \|} } } \\
\nonumber
&\;\;\;\;\;  + \sum\limits_{t = 2}^T {\sum\limits_{i = 1}^n {\frac{p}{{{\delta _t}}}} } \| {\varepsilon _{i,t - 1}^e} \| + \sum\limits_{t = 1}^{T - 2} {\sum\limits_{i = 1}^n {\sum\limits_{j = 1}^n {p\tau \| {\varepsilon _{j,t}^e} \|\sum\limits_{k = 0}^{T - t - 2} {\frac{{{\lambda ^k}}}{{{\delta _{k + t + 2}}}}} } } } \\
\nonumber
&  = \sum\limits_{t = 1}^T {\frac{{np{\varpi _1} {\lambda ^{t - 2}}}}{{{\delta _t}}} }  + \sum\limits_{t = 2}^T {\sum\limits_{i = 1}^n {\frac{{2p}}{{{\delta _t}}}\| {\varepsilon _{i,t - 1}^e} \|} } \\
\nonumber
&\;\;\;\;\;  + \sum\limits_{t = 1}^{T - 2} {\sum\limits_{i = 1}^n {np\tau \| {\varepsilon _{i,t}^e} \|\sum\limits_{k = 0}^{T - t - 2} {\frac{{{\lambda ^k}}}{{{\delta _{k + t + 2}}}}} } } \\
\nonumber
& \le \sum\limits_{t = 1}^T {\frac{{np{\varpi _1}{\lambda ^{t - 2}}}}{{{\delta _t}}}} + \sum\limits_{t = 1}^{T - 2} {\sum\limits_{i = 1}^n {np\tau \| {\varepsilon _{i,t}^e} \|\sum\limits_{t = 0}^{T - t - 2} {\frac{{{\lambda ^k}}}{{{\delta _{k + t + 2}}}}} } } \\
&\;\;\;\;\;  + \sum\limits_{t = 2}^T {\sum\limits_{i = 1}^n {( {\frac{{a{p^2}{F_1}{\alpha _t}}}{{\delta _t^2}} + \frac{1}{{a{F_1}{\alpha _t}}}{{\| {\varepsilon _{i,t - 1}^e} \|}^2}} )} }. \label{theorem1-proof-eq10}
\end{flalign}
\end{small}%
From \eqref{lemma3-eq1}, \eqref{theorem1-proof-eq8} and $a > 0$, setting ${\kappa _t} = {\delta _t}{\beta _t}$ for $t \in [ T ]$ gives\par\nobreak\vspace{-10pt}
\begin{small}
\begin{flalign}
\nonumber
&\;\;\;\;\; \sum\limits_{t = 1}^T {\sum\limits_{i = 1}^n {\frac{p{F_1}}{{{\delta _t}{\beta _t}}}\| {{e_{i,t}} - {{\bar e}_t}} \|} } \\
\nonumber
& \le \sum\limits_{t = 1}^T {\sum\limits_{i = 1}^n {\frac{{p{F_1}{\varpi _1} {\lambda ^{t - 2}}}}{{{\delta _t}{\beta _t}}} } } + \frac{1}{n}\sum\limits_{t = 2}^T {\sum\limits_{i = 1}^n {\sum\limits_{j = 1}^n {\frac{p{F_1}}{{{\delta _t}{\beta _t}}}\| {\varepsilon _{j,t - 1}^e} \|} } }   \\
\nonumber
&\;\;\;\;\; + \sum\limits_{t = 2}^T {\sum\limits_{i = 1}^n {\frac{p{F_1}}{{{\delta _t}{\beta _t}}}} } \| {\varepsilon _{i,t - 1}^e} \| \\
\nonumber
&\;\;\;\;\;  + \sum\limits_{t = 1}^{T - 2} {\sum\limits_{i = 1}^n {\sum\limits_{j = 1}^n {p\tau {F_1}\| {\varepsilon _{j,t}^e} \|\sum\limits_{s = 0}^{T - t - 2} {\frac{{{\lambda ^s}}}{{{\delta _{s + t + 2}}{\beta _{s + t + 2}}}}} } } } 
\end{flalign}
\end{small}%
\begin{small}
\begin{flalign}
\nonumber
& \le \sum\limits_{t = 1}^T {\frac{{np{F_1}{\varpi _1}{\lambda ^{t - 2}}}}{{{\delta _t}{\beta _t}}}}
  + \sum\limits_{t = 1}^{T - 2} {\sum\limits_{i = 1}^n {np\tau {F_1}\| {\varepsilon _{i,t}^e} \|\sum\limits_{k = 0}^{T - t - 2} {\frac{{{\lambda ^k}}}{{{\delta _{k + t + 2}}{\beta _{k + t + 2}}}}} } } \\
\nonumber
&\;\;\;\;\;  + \sum\limits_{t = 2}^T {\sum\limits_{i = 1}^n {\frac{{2p{F_1}}}{{{\delta _t}{\beta _t}}}\| {\varepsilon _{i,t - 1}^e} \|} } \\
\nonumber
& \le \sum\limits_{t = 1}^T {\frac{{np{F_1}{\varpi _1}{\lambda ^{t - 2}}}}{{{\delta _t}{\beta _t}}}}
+ \sum\limits_{t = 1}^{T - 2} {\sum\limits_{i = 1}^n {np\tau {F_1}\| {\varepsilon _{i,t}^e} \|\sum\limits_{k = 0}^{T - t - 2} {\frac{{{\lambda ^k}}}{{{\delta _{k + t + 2}}{\beta _{k + t + 2}}}}} } } \\
&\;\;\;\;\;  + \sum\limits_{t = 2}^T {\sum\limits_{i = 1}^n {( {\frac{{a{p^2}F_1^3{\alpha _t}}}{{\delta _t^2\beta _t^2}} + \frac{1}{{a{F_1}{\alpha _t}}}{{\| {\varepsilon _{i,t - 1}^e} \|}^2}} )} }. \label{theorem1-proof-eq11}
\end{flalign}
\end{small}%
From $a > 0$, we have\par\nobreak\vspace{-10pt}
\begin{small}
\begin{flalign}
\nonumber
&\;\;\;\;\;\frac{p}{{{\delta _t}}}\| \varepsilon _{i,t} \| \\
\nonumber
& \le \frac{p}{{{\delta _t}}}\| {{e_{i,t}} - {z_{i,t + 1}}} \| + \frac{p}{{{\delta _t}}}\| {{z_{i,t + 1}} - {e_{i,t + 1}}} \|\\
& \le \frac{p}{{{\delta _t}}}\| {{e_{i,t}} - {z_{i,t + 1}}} \| + \frac{{a{p^2}{F_1}{\alpha _{t + 1}}}}{{4\delta _t^2}} + \frac{1}{{a{F_1}{\alpha _{t + 1}}}}{\| {\varepsilon _{i,t}^e} \|^2}. \label{theorem1-proof-eq12}
\end{flalign}
\end{small}%
From $a > 0$ and \eqref{lemma4-proof-eq1}, we have
\par\nobreak\vspace{-10pt}
\begin{small}
\begin{flalign}
\nonumber
&\;\;\;\;\; \frac{p}{{{\delta _t}}}\| {{q_{i,t}}} \|\| \varepsilon _{i,t} \| \\
\nonumber
&  \le \frac{p}{{{\delta _t}}}\| {{q_{i,t}}} \|\| {{e_{i,t}} - {z_{i,t + 1}}} \| + \frac{p}{{{\delta _t}}}\| {{q_{i,t}}} \|\| {{z_{i,t + 1}} - {e_{i,t + 1}}} \| \\
\nonumber
&  \le \frac{p}{{{\delta _t}}}\| {{q_{i,t}}} \|\| {{e_{i,t}} - {z_{i,t + 1}}} \| \\
\nonumber
&\;\;\;\;\;   + \frac{p}{{{\delta _t}}}\| {{u_i}} \|\| {{z_{i,t + 1}} - {e_{i,t + 1}}} \| + \frac{p}{{{\delta _t}}}\| {{q_{i,t}} - {u_i}} \|\| {{z_{i,t + 1}} - {e_{i,t + 1}}} \| \\
\nonumber
&  \le \frac{p{F_1}}{{{\delta _t}{\beta _t}}}\| {{e_{i,t}} - {z_{i,t + 1}}} \| + \frac{{a{p^2}{F_1}{\alpha _{t + 1}}}}{{2\delta _t^2}}{\| {{\mu _i}} \|^2} + \frac{1}{{2a{F_1}{\alpha _{t + 1}}}}{\| {\varepsilon _{i,t}^e} \|^2} \\
&\;\;\;\;\;   + \frac{{a{p^2}{F_1}{\alpha _{t + 1}}}}{{2\delta _t^2}}{\| \psi _{i,t} \|^2} + \frac{1}{{2a{F_1}{\alpha _{t + 1}}}}{\| {\varepsilon _{i,t}^e} \|^2}. \label{theorem1-proof-eq13}
\end{flalign}
\end{small}%
From \eqref{Algorithm1-eq1}, we have\par\nobreak\vspace{-10pt}
\begin{small}
\begin{flalign}
\nonumber
\sum\limits_{i = 1}^n {\| {{e_{i,t}} - {z_{i,t + 1}}} \|}
&\le \sum\limits_{i = 1}^n {( {\| {{e_{i,t}} - {{\bar e}_t}} \| + \| {{{\bar e}_t} - {z_{i,t + 1}}} \|} )} \\
\nonumber
&= \sum\limits_{i = 1}^n {( {\| {{e_{i,t}} - {{\bar e}_t}} \| + \| {{{\bar e}_t} - \sum\limits_{j = 1}^n {{{[ {{W_t}} ]}_{ij}}{e_{j,t}}}} \|} )} \\
\nonumber
& \le \sum\limits_{i = 1}^n {\| {{e_{i,t}} - {{\bar e}_t}} \|}  + \sum\limits_{i = 1}^n {\sum\limits_{j = 1}^n {{{[ {{W_t}} ]}_{ij}}\| {{{\bar e}_t} - {e_{j,t}}} \|} } \\
& \le 2\sum\limits_{i = 1}^n {\| {{e_{i,t}} - {{\bar e}_t}} \|}. \label{theorem1-proof-eq14}
\end{flalign}
\end{small}%
From \eqref{theorem1-proof-eq10}, \eqref{theorem1-proof-eq12}, \eqref{theorem1-proof-eq14} and $a > 0$, we have\par\nobreak\vspace{-10pt}
\begin{small}
\begin{flalign}
\nonumber
&\;\;\;\;\; \sum\limits_{t = 1}^T {\sum\limits_{i = 1}^n {\frac{p}{{{\delta _t}}}\| \varepsilon _{i,t} \|} } \\
\nonumber
&  \le \sum\limits_{t = 1}^T {\frac{{2np{\varpi _1}{\lambda ^{t - 2}}}}{{{\delta _t}}}} + \sum\limits_{t = 1}^{T - 2} {\sum\limits_{i = 1}^n {2np\tau \| {\varepsilon _{i,t}^e} \|\sum\limits_{k = 0}^{T - t - 2} {\frac{{{\lambda ^k}}}{{{\delta _{k + t + 2}}}}} } } \\
&\;\;\;\;\;   + \sum\limits_{t = 1}^T {\frac{{9an{p^2}{F_1}{\alpha _t}}}{{4\delta _t^2}}}  + \sum\limits_{t = 1}^T {\sum\limits_{i = 1}^n {\frac{3}{{a{F_1}{\alpha _{t + 1}}}}{{\| {\varepsilon _{i,t}^e} \|}^2}} }. \label{theorem1-proof-eq15}
\end{flalign}
\end{small}%
From \eqref{theorem1-proof-eq11}, \eqref{theorem1-proof-eq13}, \eqref{theorem1-proof-eq14} and $a > 0$, we have\par\nobreak\vspace{-10pt}
\begin{small}
\begin{flalign}
\nonumber
&\;\;\;\;\; \sum\limits_{t = 1}^T {\sum\limits_{i = 1}^n {\frac{p}{{{\delta _t}}}\| {{q_{i,t}}} \|\| \varepsilon _{i,t} \|} }  \le \sum\limits_{t = 1}^T {\frac{{2np{F_1}{\varpi _1}{\lambda ^{t - 2}}}}{{{\delta _t}{\beta _t}}}}\\
\nonumber
& \;\;\;\;\;+ \sum\limits_{t = 1}^T {\sum\limits_{i = 1}^n {\frac{3}{{a{F_1}{\alpha _{t + 1}}}}{{\| {\varepsilon _{i,t}^e} \|}^2}}}  
 + \sum\limits_{t = 1}^T {\frac{{2an{p^2}F_1^3{\alpha _t}}}{{\delta _t^2\beta _t^2}}} 
\end{flalign}
\end{small}%
\begin{small}
\begin{flalign}
\nonumber
&\;\;\;\;\;  + \sum\limits_{t = 1}^T {\sum\limits_{i = 1}^n {\frac{{a{p^2}{F_1}{\alpha _t}}}{{2\delta _t^2}}{{\| {{\mu_i}} \|}^2}} }  + \sum\limits_{t = 1}^T {\sum\limits_{i = 1}^n {\frac{{a{p^2}{F_1}{\alpha _t}}}{{2\delta _t^2}}{{\| \psi _{i,t} \|}^2}} } \\
&\;\;\;\;\;  + \sum\limits_{t = 1}^{T - 2} {\sum\limits_{i = 1}^n {2np\tau {F_1}\| {\varepsilon _{i,t}^e} \|\sum\limits_{k = 0}^{T - t - 2} {\frac{{{\lambda ^k}}}{{{\delta _{k + t + 2}}{\beta _{k + t + 2}}}}} } }. \label{theorem1-proof-eq16}
\end{flalign}
\end{small}%
From \eqref{lemma4-proof-eq1}, \eqref{theorem1-proof-eq9}, \eqref{theorem1-proof-eq15} and \eqref{theorem1-proof-eq16}, we have\par\nobreak\vspace{-10pt}
\begin{small}
\begin{flalign}
\nonumber
&\;\;\;\;\; \sum\limits_{t = 1}^T {\sum\limits_{i = 1}^n {\mathbf{E}[ {{{\tilde \Delta }_{i,t + 1}}} ]} } \\
\nonumber
& \le 2n{F_2}{\varpi _2} + \sum\limits_{t = 1}^T {\frac{{2np{F_1}{\varpi _1}{\lambda ^{t - 2}}}}{{{\delta _t}}}}  + \sum\limits_{t = 1}^T {\frac{{2np{F_1^2}{\varpi _1}{\lambda ^{t - 2}}}}{{{\delta _t}{\beta _t}}}} \\
\nonumber
&\;\;\;\;\;  + \sum\limits_{t = 1}^T {2an{F_2}{\varpi _3}{\alpha _t}}  + \sum\limits_{t = 1}^T {\frac{{9an{p^2}F_1^2{\alpha _t}}}{{4\delta _t^2}}} + \sum\limits_{t = 1}^T {\frac{{2an{p^2}F_1^4{\alpha _t}}}{{\delta _t^2\beta _t^2}}} \\
\nonumber
&\;\;\;\;\;  + \sum\limits_{t = 1}^T {\sum\limits_{i = 1}^n {\frac{{a{p^2}F_1^2{\alpha _t}}}{{2\delta _t^2}}{{\| {{\mu _i}} \|}^2}} }  + \sum\limits_{t = 1}^T {\sum\limits_{i = 1}^n {\frac{{a{p^2}F_1^2{\alpha _t}}}{{2\delta _t^2}}{{\| \psi _{i,t} \|}^2}} } \\
\nonumber
&\;\;\;\;\;  + \sum\limits_{t = 1}^{T - 2} {\sum\limits_{i = 1}^n {2np\tau {F_1}\| {\varepsilon _{i,t}^e} \|\sum\limits_{k = 0}^{T - t - 2} {\frac{{{\lambda ^k}}}{{{\delta _{k + t + 2}}}}} } } \\
\nonumber
&\;\;\;\;\;  + \sum\limits_{t = 1}^{T - 2} {\sum\limits_{i = 1}^n {2np\tau {F_1^2}\| {\varepsilon _{i,t}^e} \|\sum\limits_{k = 0}^{T - t - 2} {\frac{{{\lambda ^k}}}{{{\delta _{k + t + 2}}{\beta _{k + t + 2}}}}} } } \\
\nonumber
&\;\;\;\;\;  + n{F_1}{F_2}\sum\limits_{t = 1}^T \frac{{\psi _t^1 + {\delta _t}}}{{{\beta _t}}} + n{F_2}\sum\limits_{t = 1}^T {( \psi _t^1 + 5{\delta _t} )} \\
&\;\;\;\;\; + \sum\limits_{t = 1}^T {\sum\limits_{i = 1}^n {\frac{{16 - a}}{{2a{\alpha _{t + 1}}}}} } {\| {\varepsilon _{i,t}^e} \|^2}. \label{theorem1-proof-eq17}
\end{flalign}
\end{small}%

($\mathbf{i}\mbox{-}\mathbf{3}$)
Finally, we are ready to prove \eqref{theorem1-eq2}.

Combining \eqref{theorem1-proof-eq5}, \eqref{theorem1-proof-eq7}, and \eqref{theorem1-proof-eq17}, choosing ${\mu _i} = {\mathbf{0}_{{m_i}}}$, and noting that ${{\gamma _t}}$ is non-increasing, we have
\par\nobreak\vspace{-10pt}
\begin{small}
\begin{flalign}
\nonumber
&{\rm{\mathbf{E}}}[ {{\rm{Net}\mbox{-}\rm{Reg}}( {\{ {{x_{i,t}}} \},{y_{[ T ]}}} )} ]\\
\nonumber
& \le 2{F_2}{\varpi _2} + 2F_1^2\sum\limits_{t = 1}^T {{\gamma _t}}  + \sum\limits_{t = 1}^T {\frac{{2p{F_1}{\varpi _1}{\lambda ^{t - 2}}}}{{{\delta _t}}}}  + \sum\limits_{t = 1}^T {\frac{{2pF_1^2{\varpi _1}{\lambda ^{t - 2}}}}{{{\delta _t}{\beta _t}}}} \\
\nonumber
&  + \sum\limits_{t = 1}^T {2a{F_2}{\varpi_3}{\alpha _t}}  + \sum\limits_{t = 1}^T {\frac{{9a{p^2}F_1^2{\alpha _t}}}{{4\delta _t^2}}}  + \sum\limits_{t = 1}^T {\frac{{2a{p^2}F_1^4{\alpha _t}}}{{\delta _t^2\beta _t^2}}} \\
\nonumber
&  + \sum\limits_{t = 1}^{T - 2} {\sum\limits_{i = 1}^n {2p\tau {F_1}\| {\varepsilon _{i,t}^e} \|\sum\limits_{k = 0}^{T - t - 2} {\frac{{{\lambda ^k}}}{{{\delta _{k + t + 2}}}}} } } \\
\nonumber
&  + \sum\limits_{t = 1}^{T - 2} {\sum\limits_{i = 1}^n {2p\tau {F_1^2}\| {\varepsilon _{i,t}^e} \|\sum\limits_{k = 0}^{T - t - 2} {\frac{{{\lambda ^k}}}{{{\delta _{k + t + 2}}{\beta _{k + t + 2}}}}} } } \\
\nonumber
&  + {F_1}{F_2}\sum\limits_{t = 1}^T \frac{{\psi _t^1 + {\delta _t}}}{{{\beta _t}}} + {F_2}\sum\limits_{t = 1}^T {( \psi _t^1 + 5{\delta _t} )} \\
\nonumber
&  + \frac{1}{n}\sum\limits_{t = 1}^T {\sum\limits_{i = 1}^n {\frac{{16 - a}}{{2a{\alpha _{t + 1}}}}} } {\| {\varepsilon _{i,t}^e} \|^2} \\
\nonumber
& + \frac{{2R{{( \mathbb{X} )}^2}}}{{{\alpha _{T + 1}}}} + \frac{{2R{{( \mathbb{X} )}}}}{{{\alpha _T}}}{P_T} + \sum\limits_{t = 1}^T {\frac{{2R{{( \mathbb{X} )}^2}( {{\xi _t} - {\xi _{t + 1}}} )}}{{{\alpha _{t + 1}}}}}\\
& + \frac{1}{{2n}}\sum\limits_{t = 1}^T {\sum\limits_{i = 1}^n {( {\frac{{a{p^2}F_1^2{\alpha _t}}}{{\delta _t^2}} - \psi _t^2} )} } \mathbf{E}[ {{{\| {{q_{i,t}}} \|}^2}} ].
\label{theorem1-proof-eq17.5}
\end{flalign}
\end{small}%
It follows from (18) that\par\nobreak\vspace{-10pt}
\begin{small}
\begin{flalign}
\nonumber
&\;\;\;\;\; \sum\limits_{t = 1}^{T - 2} {\sum\limits_{i = 1}^n {2p\tau {F_1}\| {\varepsilon _{i,t}^e} \|\sum\limits_{k = 0}^{T - t - 2} {\frac{{{\lambda ^k}}}{{{\delta _{k + t + 2}}}}} } } \\
\nonumber
&  = \sum\limits_{t = 1}^{T - 2} {\sum\limits_{i = 1}^n {2p\tau {F_1}\| {\varepsilon _{i,t}^e} \|\sum\limits_{k = 0}^{T - t - 2} {\frac{{{{(k + t + 3)}^{{g _3}}}{\lambda ^k}}}{{r( \mathbb{X} )}}} } } \\
\nonumber
& \le \sum\limits_{t = 1}^{T - 2} {\sum\limits_{i = 1}^n {2p\tau {F_1}} } {( {T + 1} )^{{g_3}}}\| {\varepsilon _{i,t}^e} \|\sum\limits_{k = 0}^{T - t - 2} {\frac{{{\lambda ^k}}}{{r( \mathbb{X} )}}} \\
\nonumber
& \le \sum\limits_{t = 1}^{T - 2} {\sum\limits_{i = 1}^n {\frac{{2p\tau {F_1}{{( {T + 1} )}^{{g_3}}}}}{{( {1 - \lambda } )r( \mathbb{X} )}}} } \| {\varepsilon _{i,t}^e} \| \\
& \le \sum\limits_{t = 1}^{T - 2} {\sum\limits_{i = 1}^n {\big( {\frac{{an{p^2}{\tau ^2}F_1^2{{( {T + 1} )}^{2{g_3}}}{\alpha _t}}}{{{{( {1 - \lambda } )}^2}r{{( \mathbb{X} )}^2}}} + \frac{1}{{an{\alpha _t}}}{{\| {\varepsilon _{i,t}^e} \|}^2}} \big)} }, \label{theorem1-proof-eq18}
\end{flalign}
\end{small}%
and
\par\nobreak\vspace{-10pt}
\begin{small}
\begin{flalign}
\nonumber
&\;\;\;\;\; \sum\limits_{t = 1}^{T - 2} {\sum\limits_{i = 1}^n {2p\tau {F_1^2}\| {\varepsilon _{i,t}^e} \|\sum\limits_{k = 0}^{T - t - 2} {\frac{{{\lambda ^k}}}{{{\delta _{k + t + 2}}{\beta _{k + t + 2}}}}} } } \\
\nonumber
& \le \sum\limits_{t = 1}^{T - 2} {\sum\limits_{i = 1}^n {2p\tau {F_1^2}\| {\varepsilon _{i,t}^e} \|\sum\limits_{k = 0}^{T - t - 2} {\frac{{{{(k + t + 3)}^{{g _2} + {g _3}}}{\lambda ^k}}}{{2r( \mathbb{X} )}}} } } \\
\nonumber
& \le \sum\limits_{t = 1}^{T - 2} {\sum\limits_{i = 1}^n {2p\tau F_1^2} } {( {T + 1} )^{{g_2} + {g_3}}}\| {\varepsilon _{i,t}^e} \|\sum\limits_{k = 0}^{T - t - 2} {\frac{{{\lambda ^k}}}{{2r( \mathbb{X} )}}} \\
\nonumber
& \le \sum\limits_{t = 1}^{T - 2} {\sum\limits_{i = 1}^n {\frac{{p\tau F_1^2{{( {T + 1} )}^{{g_2} + {g_3}}}}}{{( {1 - \lambda } )r( \mathbb{X} )}}} } \| {\varepsilon _{i,t}^e} \| \\
& \le \sum\limits_{t = 1}^{T - 2} {\sum\limits_{i = 1}^n {\big( {\frac{{an{p^2}{\tau ^2}F_1^4{{( {T + 1} )}^{2{g_2} + 2{g_3}}}{\alpha _t}}}{{4{{( {1 - \lambda } )}^2}r{{( \mathbb{X} )}^2}}} + \frac{1}{{an{\alpha _t}}}{{\| {\varepsilon _{i,t}^e} \|}^2}} \big)} }. \label{theorem1-proof-eq19}
\end{flalign}
\end{small}%
For any $g  \in [ {0,1} )$ and $T \in {\mathbb{N}_ + }$, we have\par\nobreak\vspace{-10pt}
\begin{small}
\begin{flalign}
\sum\limits_{t = 1}^T {\frac{1}{{{t^g }}}}  \le \int_1^T {\frac{1}{{{t^g }}}} dt + 1 = \frac{{{T^{1 - g }} - g }}{{1 - g }} \le \frac{{{T^{1 - g }}}}{{1 - g }}. \label{theorem1-proof-eq20}
\end{flalign}
\end{small}%
According to Lemma~4 in \cite{Yi2021b}, for any constants $\theta  \in [ {0,1} ]$ and $g \in [ {0,1} )$, it holds that
\par\nobreak\vspace{-10pt}
\begin{small}
\begin{flalign}
{( {t + 1} )^g}( {\frac{1}{{{t^\theta }}} - \frac{1}{{{{( {t + 1} )}^\theta }}}} ) \le \frac{1}{t}, \forall t \in {\mathbb{N}_ + }. \label{theorem1-proof-eq20.1}
\end{flalign}
\end{small}%
For any $T \ge 3$, we have\par\nobreak\vspace{-10pt}
\begin{small}
\begin{flalign}
\sum\limits_{t = 1}^T {\frac{1}{t}}  \le 1 + \int_1^T {\frac{1}{t}} dt \le 2\log ( T ). \label{theorem1-proof-eq21}
\end{flalign}
\end{small}%
From \eqref{theorem1-eq1} and ${g _1} - 2{g _3} \ge {g _2}$, when $a \in ( {0,20} ]$, we have\par\nobreak\vspace{-10pt}
\begin{small}
\begin{flalign}
\nonumber
&\;\;\;\;\; \frac{{a{p^2}F_1^2{\alpha _t}}}{{\delta _t^2}} - \psi _t^2 \\
\nonumber
& = \frac{{a{p^2}F_1^2{\alpha _t}}}{{\delta _t^2}} + \frac{1}{{{\gamma _{t + 1}}}} - \frac{1}{{{\gamma _t}}} - {\beta _{t + 1}} \\
\nonumber
&  = \frac{a}{{20{{( {t + 1} )}^{{g _1} - 2{g _3}}}}} + \frac{{t + 1}}{{{{( {t + 1} )}^{{g _2}}}}} - \frac{t}{{{t^{{g _2}}}}} - \frac{2}{{{{( {t + 1} )}^{{g _2}}}}} \\
\nonumber
& \le \frac{a}{{20{{( {t + 1} )}^{{g _2}}}}} + \frac{{t + 1}}{{{{( {t + 1} )}^{{g _2}}}}} - \frac{t}{{{t^{{g _2}}}}} - \frac{2}{{{{( {t + 1} )}^{{g _2}}}}} \\
\nonumber
& \le \frac{1}{{{{( {t + 1} )}^{{g _2}}}}} + \frac{{t + 1}}{{{{( {t + 1} )}^{{g _2}}}}} - \frac{t}{{{t^{{g _2}}}}} - \frac{2}{{{{( {t + 1} )}^{{g _2}}}}} \\
& \le \frac{t}{{{{( {t + 1} )}^{{g _2}}}}} - \frac{t}{{{t^{{g _2}}}}} \le 0. \label{theorem1-proof-eq22}
\end{flalign}
\end{small}%
For any $g \in [ {0,1} ]$ and $T \in {\mathbb{N}_ + }$, we have
\begin{small}
\begin{flalign}
{( {T + 1} )^g} \le 2{T^g}. \label{theorem1-proof-eq22.1}
\end{flalign}
\end{small}%
Combining \eqref{theorem1-eq1}, \eqref{theorem1-proof-eq17.5}--\eqref{theorem1-proof-eq22.1}, and choosing $a = 20$ yields 
\par\nobreak\vspace{-10pt}
\begin{small}
\begin{flalign}
\nonumber
&{\mathbf{E}}[ {{\rm{Net}\mbox{-}\rm{Reg}}( {\{ {{x_{i,t}}} \},{y_{[T]}}} )} ]\\
\nonumber
& \le 2{F_2}{\varpi _2} + \frac{{2F_1^2{T^{{g _2}}}}}{{{g _2}}} + \frac{{4p{F_1}{\varpi _1}{T^{{g_3}}}}}{{\lambda ( {1 - \lambda } )r( \mathbb{X} )}} + \frac{{2pF_1^2{\varpi _1}{T^{{g_2} + {g_3}}}}}{{\lambda ( {1 - \lambda } )r( \mathbb{X} )}}\\
\nonumber
&  + \frac{{2{F_2}{\varpi _3}r{{( \mathbb{X} )}^2}{T^{1 - {g_1}}}}}{{{p^2}F_1^2( {1 - {g_1}} )}} + \frac{{9{T^{1 - {g_1} + {g_3}}}}}{{4( {1 - {g_1} + {g_3}} )}} + \frac{{F_1^2{T^{1 - {g_1} + 2{g_2} + 2{g_3}}}}}{{2( {1 - {g_1} + 2{g_2} + 2{g_3}} )}} \\
\nonumber
&  + \frac{{2{n^2}{\tau ^2}{T^{1 - {g_1} + 2{g_3}}}}}{{{{( {1 - \lambda } )}^2}( {1 - {g_1}} )}} + \frac{{{n^2}{\tau ^2}F_1^2{T^{1 - {g_1} + 2{g_2} + 2{g_3}}}}}{{2{{( {1 - \lambda } )}^2}( {1 - {g_1}} )}} \\
\nonumber
&  + \frac{{{F_1}{F_2}\big( {R( \mathbb{X} ) + 2r( \mathbb{X} )} \big){T^{1 + {g_2} - {g_3}}}}}{{2( {1 + {g_2} - {g_3}} )}} + \frac{{{F_2}\big( {R( \mathbb{X} ) + 6r( \mathbb{X} )} \big){T^{1 - {g_3}}}}}{{1 - {g_3}}}\\
\nonumber
&  + \frac{{320{p^2}F_1^2R{{( \mathbb{X} )}^2}{T^{{g_1}}}}}{{r{{( \mathbb{X} )}^2}}} + \frac{{80{p^2}F_1^2R{{( \mathbb{X} )}^2}{T^{{g_1}}}{P_T}}}{{r{{( \mathbb{X} )}^2}}} \\
&  + \frac{{80{p^2}F_1^2R{{( \mathbb{X} )}^2}\log ( T )}}{{r{{( \mathbb{X} )}^2}}}. \label{theorem1-proof-regret}
\end{flalign}
\end{small}%
Therefore, we know that \eqref{theorem1-eq2} holds.

($\mathbf{ii}$) Let us then establish the network CCV bound.

We have\par\nobreak\vspace{-10pt}
\begin{small}
\begin{flalign}
\nonumber
&\;\;\;\;\; \mu _i^T{[ {{c_{i,t}}( {{x_{i,t}}} )} ]_ + } \\
\nonumber
& = \mu _i^T{[ {{c_{i,t}}( {{x_{j,t}}} )} ]_ + } + \mu _i^T{[ {{c_{i,t}}( {{x_{i,t}}} )} ]_ + } - \mu _i^T{[ {{c_{i,t}}( {{x_{j,t}}} )} ]_ + } \\
\nonumber
& \ge \mu _i^T{[ {{c_{i,t}}( {{x_{j,t}}} )} ]_ + } - \| {{\mu _i}} \|\| {{{[ {{c_{i,t}}( {{x_{i,t}}} )} ]}_ + } - {{[ {{c_{i,t}}( {{x_{j,t}}} )} ]}_ + }} \| \\
\nonumber
& \ge \mu _i^T{[ {{c_{i,t}}( {{x_{j,t}}} )} ]_ + } - \| {{\mu _i}} \|\| {{c_{i,t}}( {{x_{i,t}}} ) - {c_{i,t}}( {{x_{j,t}}} )} \| \\
\nonumber
& \ge \mu _i^T{[ {{c_{i,t}}( {{x_{j,t}}} )} ]_ + } - {F_2}\| {{\mu _i}} \|\| {{x_{i,t}} - {x_{j,t}}} \| \\
\nonumber
& \ge \mu _i^T{[ {{c_{i,t}}( {{x_{j,t}}} )} ]_ + } \\
\nonumber
&\;\;\;\;\; - {F_2}\| {{\mu _i}} \|( \| {{x_{i,t}} - \frac{1}{n}\sum\limits_{m = 1}^n {{x_{m,t}}} } \|  + \| {{x_{j,t}} -  \frac{1}{n}\sum\limits_{m = 1}^n {{x_{m,t}}} } \| ) \\
\nonumber
& = \mu _i^T{[ {{c_{i,t}}( {{x_{j,t}}} )} ]_ + } - {F_2}\| {{\mu _i}} \|\| {{e_{i,t}} - {{\bar e}_t} + \frac{{{\delta _t}}}{n}\sum\limits_{m = 1}^n {( {{u_{i,t}} - {u_{m,t}}} )} } \| \\
\nonumber
&\;\;\;\;\; - {F_2}\| {{\mu _i}} \|\| {{e_{j,t}} - {{\bar e}_t} + \frac{{{\delta _t}}}{n}\sum\limits_{m = 1}^n {( {{u_{j,t}} - {u_{m,t}}} )} } \| \\
\nonumber
& \ge \mu _i^T{[ {{c_{i,t}}( {{x_{j,t}}} )} ]_ + } - {F_2}\| {{\mu _i}} \|( {\| {{e_{i,t}} - {{\bar e}_t}} \| + \| {{e_{j,t}} - {{\bar e}_t}} \|} ) \\
&\;\;\;\;\;- 4{F_2}{\delta _t}\| {{\mu _i}} \|. \label{theorem1-proof-eq23}
\end{flalign}
\end{small}%
where the second inequality since the projection ${[  \cdot  ]_ + }$ is nonexpansive, the third inequality due to \eqref{lemma5-proof-eq1b}, the last equality due to \eqref{Algorithm1-eq4}, and the last inequality due to $\| {{u_{i,t}}} \| = 1$.

Combining \eqref{theorem1-proof-eq3} and \eqref{theorem1-proof-eq23} with ${y_t} = y$, and adding up across $j \in [ n ]$ gives\par\nobreak\vspace{-10pt}
\begin{small}
\begin{flalign}
\nonumber
&\;\;\;\;\; \sum\limits_{i = 1}^n {\big( {{\Delta _{i,t + 1}}( {{\mu _i}} ) + \frac{1}{n}\sum\limits_{j = 1}^n {\mu _i^T{{[ {{c_{i,t}}( {{x_{j,t}}} )} ]}_ + }}  - \frac{1}{2}{\beta _{t + 1}}{{\| {{\mu _i}} \|}^2}} \big)} \\
\nonumber
&\;\;\;\;\; + \sum\limits_{i = 1}^n {{l_t}( {{x_{i,t}}} )}  - n{l_t}( y ) \\
\nonumber
& \le 2nF_1^2{\gamma _{t + 1}} + \sum\limits_{i = 1}^n {4n{F_2}{\delta _t}\| {{\mu _i}} \|} + \sum\limits_{i = 1}^n {{{\hat \Delta }_{i,t + 1}}( {{\mu _i}} ) + \frac{1}{n}} {\check{\Delta} _t}( {{\mu _i}} )\\
\nonumber
&\;\;\;\;\;  + \frac{1}{{2{\alpha _{t + 1}}}}\sum\limits_{i = 1}^n {{\mathbf{E}_{{\mathfrak{U}_t}}}[ {{{\| {{\hat y}_t - {z_{i,t + 1}}} \|}^2} - {{\| {{\hat y}_{t + 1} - {z_{i,t + 2}}} \|}^2}} ]} \\
&\;\;\;\;\; + \frac{{2nR{{( \mathbb{X} )}^2}}}{{{\alpha _{t + 1}}}}{\mathbf{E}_{{\mathfrak{U}_t}}}[ {( {{\xi _t} - {\xi _{t + 1}}} )} ], \label{theorem1-proof-eq24}
\end{flalign}
\end{small}%
where\par\nobreak\vspace{-10pt}
\begin{small}
\begin{flalign}
\nonumber
{\hat \Delta _{i,t + 1}}( {{\mu _i}} ) &= {F_2}\| {{\mu _i}} \|\| {{e_{i,t}} - {{\bar e}_t}} \| + {\tilde \Delta _{i,t + 1}},\\
\nonumber
{\check{\Delta} _t}( {{\mu _i}} ) &= \sum\limits_{i = 1}^n {\sum\limits_{j = 1}^n {{F_2}} } \| {{\mu _i}} \|\| {{e_{j,t}} - {{\bar e}_t}} \|.
\end{flalign}
\end{small}%

To obtain \eqref{theorem1-eq3}, next we derive the upper bounds for $\sum\nolimits_{t = 1}^T {\sum\nolimits_{i = 1}^n {\mathbf{E}[ {{{\hat \Delta }_{i,t + 1}}}( {{\mu _i}} ) ]} }$ and $\frac{1}{n}\sum\nolimits_{t = 1}^T {\mathbf{E}[ {{\check{\Delta} _t}}( {{\mu _i}} ) ]}$, respectively.

($\mathbf{ii}\mbox{-}\mathbf{1}$)
From \eqref{lemma3-eq1}, \eqref{theorem1-proof-eq8} and $a > 0$, setting ${\kappa _t} = 1$ for $t \in [ T ]$ gives\par\nobreak\vspace{-10pt}
\begin{small}
\begin{flalign}
\nonumber
&\;\;\;\;\; \sum\limits_{t = 1}^T {\sum\limits_{i = 1}^n {\| {{\mu _i}} \|\| {{e_{i,t}} - {{\bar e}_t}} \|} } \\
\nonumber
& \le \sum\limits_{t = 1}^T {\sum\limits_{i = 1}^n {{\varpi _1} {\lambda ^{t - 2}}\| {{\mu _i}} \| } }  + \frac{1}{n}\sum\limits_{t = 2}^T {\sum\limits_{i = 1}^n {\sum\limits_{j = 1}^n {\| {{\mu _i}} \|\| {\varepsilon _{j,t - 1}^e} \|} } } \\
\nonumber
&  + \sum\limits_{t = 2}^T {\sum\limits_{i = 1}^n {\| {{\mu _i}} \|\| {\varepsilon _{i,t - 1}^e} \|} }  + \sum\limits_{t = 1}^{T - 2} {\sum\limits_{i = 1}^n {\sum\limits_{j = 1}^n {\frac{\tau }{{1 - \lambda }}\| {{\mu _i}} \|\| {\varepsilon _{j,t}^e} \|} } } \\
\nonumber
& \le {\varpi _2}\sum\limits_{i = 1}^n {\| {{\mu _i}} \|} + \frac{1}{n}\sum\limits_{t = 2}^T {\sum\limits_{i = 1}^n {\sum\limits_{j = 1}^n {\| {\varepsilon _{i,t - 1}^e} \|\| {{\mu _j}} \|} } } \\
\nonumber
&\;\;\;\;\;  + \sum\limits_{t = 2}^T {\sum\limits_{i = 1}^n {\| {\varepsilon _{i,t - 1}^e} \|\| {{\mu _i}} \|} }  + \sum\limits_{t = 2}^T {\sum\limits_{i = 1}^n {\sum\limits_{j = 1}^n {\frac{\tau }{{1 - \lambda }}\| {\varepsilon _{i,t - 1}^e} \|\| {{\mu _j}} \|} } } \\
\nonumber
& \le {\varpi _2}\sum\limits_{i = 1}^n {\| {{\mu _i}} \|}\\
\nonumber
&\;\;\;\;\;  + \frac{1}{n}\sum\limits_{t = 2}^T {\sum\limits_{i = 1}^n {\sum\limits_{j = 1}^n {( {\frac{1}{{4a{F_2}{\alpha _t}}}{{\| {\varepsilon _{i,t - 1}^e} \|}^2} + a{F_2}{\alpha _t}{{\| {{\mu _j}} \|}^2}} )} } } \\
\nonumber
&\;\;\;\;\;  + \sum\limits_{t = 2}^T {\sum\limits_{i = 1}^n {( {\frac{1}{{4a{F_2}{\alpha _t}}}{{\| {\varepsilon _{i,t - 1}^e} \|}^2} + a{F_2}{\alpha _t}{{\| {{\mu _i}} \|}^2}} )} } \\
\nonumber
&\;\;\;\;\;  + \sum\limits_{t = 2}^T {\sum\limits_{i = 1}^n {\sum\limits_{j = 1}^n {( {\frac{1}{{2an{F_2}{\alpha _t}}}{{\| {\varepsilon _{i,t - 1}^e} \|}^2} + \frac{{an{\tau ^2}{F_2}{\alpha _t}}}{{2{{( {1 - \lambda } )}^2}}}{{\| {{\mu _j}} \|}^2}} )} } } \\
& \le {\varpi _2}\sum\limits_{i = 1}^n {\| {{\mu _i}} \|} + \sum\limits_{t = 2}^T {\sum\limits_{i = 1}^n {( {\frac{1}{{a{F_2}{\alpha _t}}}{{\| {\varepsilon _{i,t - 1}^e} \|}^2} + a{\varpi _3}{\alpha _t}{{\| {{\mu _i}} \|}^2}} )} }. \label{theorem1-proof-eq25}
\end{flalign}
\end{small}%
Combining \eqref{theorem1-proof-eq17} and \eqref{theorem1-proof-eq25} yields\par\nobreak\vspace{-10pt}
\begin{small}
\begin{flalign}
\nonumber
&\;\;\;\;\; \sum\limits_{t = 1}^T {\sum\limits_{i = 1}^n {{\bf{E}}[ {{{\hat \Delta }_{i,t + 1}}( {{\mu _i}} )} ]} } \\
\nonumber
& \le 2n{F_2}{\varpi _2} + {F_2}{\varpi _2}\sum\limits_{i = 1}^n {\| {{\mu _i}} \|}  + \sum\limits_{t = 1}^T {\frac{{2np{F_1}{\varpi _1}{\lambda ^{t - 2}}}}{{{\delta _t}}}} \\
\nonumber
&\;\;\;\;\;  + \sum\limits_{t = 1}^T {\frac{{2np{F_1^2}{\varpi _1}{\lambda ^{t - 2}}}}{{{\delta _t}{\beta _t}}}} + \sum\limits_{t = 1}^T {2an{F_2}{\varpi _3}{\alpha _t}}  + \sum\limits_{t = 1}^T {\frac{{9an{p^2}F_1^2{\alpha _t}}}{{4\delta _t^2}}} \\
\nonumber
&\;\;\;\;\; + \sum\limits_{t = 1}^T {\frac{{2an{p^2}F_1^4{\alpha _t}}}{{\delta _t^2\beta _t^2}}} + \sum\limits_{t = 1}^T {\sum\limits_{i = 1}^n {a{F_2}{\varpi _3}{\alpha _t}{{\| {{\mu _i}} \|}^2}} } \\
\nonumber
&\;\;\;\;\; + \sum\limits_{t = 1}^T {\sum\limits_{i = 1}^n {\frac{{a{p^2}F_1^2{\alpha _t}}}{{2\delta _t^2}}{{\| {{\mu _i}} \|}^2}} } + \sum\limits_{t = 1}^T {\sum\limits_{i = 1}^n {\frac{{a{p^2}F_1^2{\alpha _t}}}{{2\delta _t^2}}{{\| \psi _{i,t} \|}^2}} } \\
\nonumber
&\;\;\;\;\; + \sum\limits_{t = 1}^{T - 2} {\sum\limits_{i = 1}^n {2np\tau {F_1}\| {\varepsilon _{i,t}^e} \|\sum\limits_{k = 0}^{T - t - 2} {\frac{{{\lambda ^k}}}{{{\delta _{k + t + 2}}}}} } } 
\end{flalign}
\end{small}%
\begin{small}
\begin{flalign}
\nonumber
&  + \sum\limits_{t = 1}^{T - 2} {\sum\limits_{i = 1}^n {2np\tau {F_1^2}\| {\varepsilon _{i,t}^e} \|\sum\limits_{k = 0}^{T - t - 2} {\frac{{{\lambda ^k}}}{{{\delta _{k + t + 2}}{\beta _{k + t + 2}}}}} } } \\
\nonumber
& + n{F_1}{F_2}\sum\limits_{t = 1}^T \frac{{\psi _t^1 + {\delta _t}}}{{{\beta _t}}} + n{F_2}\sum\limits_{t = 1}^T {( \psi _t^1 + 5{\delta _t} )} \\
& + \sum\limits_{t = 1}^T {\sum\limits_{i = 1}^n {\frac{{18 - a}}{{2a{\alpha _{t + 1}}}}} } {\| {\varepsilon _{i,t}^e} \|^2}. \label{theorem1-proof-eq26}
\end{flalign}
\end{small}%

($\mathbf{ii}\mbox{-}\mathbf{2}$)
From \eqref{lemma3-eq1}, \eqref{theorem1-proof-eq8} and $a > 0$, for ${\mu _j} \in {\mathbb{R}^{{m_j}}}$, setting ${\kappa _t} = 1$ for $t \in [ T ]$ gives\par\nobreak\vspace{-10pt}
\begin{small}
\begin{flalign}
\nonumber
&\;\;\;\;\; \sum\limits_{t = 1}^T {\sum\limits_{i = 1}^n {\sum\limits_{j = 1}^n {\| {{\mu _i}} \|\| {{e_{j,t}} - {{\bar e}_t}} \|} } }
= \sum\limits_{t = 1}^T {\sum\limits_{i = 1}^n {\sum\limits_{j = 1}^n {\| {{e_{i,t}} - {{\bar e}_t}} \|\| {{\mu _j}} \|} } } \\
\nonumber
& \le n{\varpi _2}\sum\limits_{i = 1}^n {\| {{\mu _i}} \|}  + 2\sum\limits_{t = 2}^T {\sum\limits_{i = 1}^n {\sum\limits_{j = 1}^n {\| {\varepsilon _{i,t - 1}^e} \|} } } \| {{\mu _j}} \| \\
\nonumber
&\;\;\;\;\; + \sum\limits_{t = 1}^{T - 2} {\sum\limits_{i = 1}^n {\sum\limits_{j = 1}^n {\frac{{n\tau }}{{1 - \lambda }}\| {\varepsilon _{i,t}^e} \|\| {{\mu _j}} \|} } } \\
\nonumber
& \le n{\varpi _2}\sum\limits_{i = 1}^n {\| {{\mu _i}} \|} \\
\nonumber
&\;\;\;\;\; + 2\sum\limits_{t = 2}^T {\sum\limits_{i = 1}^n {\sum\limits_{j = 1}^n {\| {\varepsilon _{i,t - 1}^e} \|} } } \| {{\mu _j}} \| + \sum\limits_{t = 2}^T {\sum\limits_{i = 1}^n {\sum\limits_{j = 1}^n {\frac{{n\tau }}{{1 - \lambda }}\| {\varepsilon _{i,t - 1}^e} \|\| {{\mu _j}} \|} } } \\
\nonumber
& \le n{\varpi _2}\sum\limits_{i = 1}^n {\| {{\mu _i}} \|} \\
\nonumber
&\;\;\;\;\; + \sum\limits_{t = 2}^T {\sum\limits_{i = 1}^n {\sum\limits_{j = 1}^n {( {\frac{1}{{2a{F_2}{\alpha _t}}}{{\| {\varepsilon _{i,t - 1}^e} \|}^2} + 2a{F_2}{\alpha _t}{{\| {{\mu _j}} \|}^2}} )} } } \\
\nonumber
&\;\;\;\;\; + \sum\limits_{t = 2}^T {\sum\limits_{i = 1}^n {\sum\limits_{j = 1}^n {( {\frac{1}{{2a{F_2}{\alpha _t}}}{{\| {\varepsilon _{i,t - 1}^e} \|}^2} + \frac{{a{n^2}{\tau ^2}{F_2}{\alpha _t}}}{{2{{( {1 - \lambda } )}^2}}}{{\| {{\mu _j}} \|}^2}} )} } } \\
& \le n{\varpi _2}\sum\limits_{i = 1}^n {\| {{\mu _i}} \|} +\sum\limits_{t = 2}^T {\sum\limits_{i = 1}^n {( {an{\varpi _3}{\alpha _t}{{\| {{\mu _i}} \|}^2} + \frac{n}{{a{F_2}{\alpha _t}}}{{\| {\varepsilon _{i,t - 1}^e} \|}^2}} )} }. \label{theorem1-proof-eq27}
\end{flalign}
\end{small}%

From \eqref{theorem1-proof-eq27}, we have\par\nobreak\vspace{-10pt}
\begin{small}
\begin{flalign}
\nonumber
{\frac{1}{n}} \sum\limits_{t = 1}^T {{\check{\Delta} _t}( {{\mu _i}} )}
& \le {F_2}{\varpi _2}\sum\limits_{i = 1}^n {\| {{\mu _i}} \|} \\
&\;\;\;\;\; + \sum\limits_{t = 2}^T {\sum\limits_{i = 1}^n {( {a{F_2}{\varpi _3}{\alpha _t}{{\| {{\mu _i}} \|}^2} + \frac{1}{{a{\alpha _t}}}{{\| {\varepsilon _{i,t - 1}^e} \|}^2}} )} }. \label{theorem1-proof-eq28}
\end{flalign}
\end{small}%

($\mathbf{ii}\mbox{-}\mathbf{3}$)
Finally, we are ready to prove \eqref{theorem1-eq3}.

We have
\par\nobreak\vspace{-10pt}
\begin{small}
\begin{flalign}
\sum\limits_{t = 1}^T {\sum\limits_{i = 1}^n {4n{F_2}{\delta _t}\| {{\mu _i}} \|} }  \le \sum\limits_{t = 1}^T {\sum\limits_{i = 1}^n {( {4{n^2}F_2^2\delta _t + {{\| {{\mu _i}} \|}^2}\delta _t} )} }. \label{theorem1-proof-eq28.5}
\end{flalign}
\end{small}%

Define ${h_{ij}}:\mathbb{R}_ + ^{{m_i}} \to \mathbb{R}$ as\par\nobreak\vspace{-10pt}
\begin{small}
\begin{flalign}
\nonumber
{h_{ij}}( {{\mu _i}} ) &= \mu _i^T{\sum\limits_{t = 1}^T {[ {{c_{i,t}}( {{x_{j,t}}} )} ]} _ + }\\
\nonumber
& - \frac{1}{2}{\| {{\mu _i}} \|^2}\Big( {\frac{1}{{{\gamma _1}}} + \sum\limits_{t = 1}^T {( {2\delta _t + {\beta _t} + 4a{F_2}{\varpi _3}{\alpha _t} + \frac{{a{p^2}F_1^2{\alpha _t}}}{{\delta _t^2}}} )} } \Big).
\end{flalign}
\end{small}%

Then, using \eqref{theorem1-proof-eq4}, \eqref{theorem1-proof-eq6}, \eqref{theorem1-proof-eq26}, \eqref{theorem1-proof-eq28}, and \eqref{theorem1-proof-eq28.5}, noting that $\{ {{\beta _t}} \}$ and $\{ {{\gamma _t}} \}$ are non-increasing, and summing \eqref{theorem1-proof-eq24} over $t \in [ T ]$ gives\par\nobreak\vspace{-10pt}
\begin{small}
\begin{flalign}
\nonumber
&\;\;\;\;\; \frac{1}{n}\sum\limits_{i = 1}^n \sum\limits_{j = 1}^n \mathbf{E}[ {{{h_{ij}}( {{\mu _i}} )} }] 
+ n\mathbf{E}[ {{\rm{Net}\mbox{-}\rm{Reg}}( {\{ {{x_{i,t}}} \},\{ y \}} )} ]\\
\nonumber
& \le 2nF_1^2\sum\limits_{t = 1}^T {{\gamma _{t}}}  + 4n^2{F_2^2}\sum\limits_{t = 1}^T {{\delta _t}}  + 2n{F_2}{\varpi _2} + 2{F_2}{\varpi _2}\sum\limits_{i = 1}^n {\| {{\mu _i}} \|} \\
\nonumber
&\;\;\;\;\;  + \sum\limits_{t = 1}^T {\frac{{2np{F_1}{\varpi _1}{\lambda ^{t - 2}}}}{{{\delta _t}}}} + \sum\limits_{t = 1}^T {\frac{{2np{F_1^2}{\varpi _1}{\lambda ^{t - 2}}}}{{{\delta _t}{\beta _t}}}} \\
\nonumber
&\;\;\;\;\;  + \sum\limits_{t = 1}^T {2an{F_2}{\varpi _3}{\alpha _t}}  + \sum\limits_{t = 1}^T {\frac{{9an{p^2}F_1^2{\alpha _t}}}{{4\delta _t^2}}} + \sum\limits_{t = 1}^T {\frac{{2an{p^2}F_1^4{\alpha _t}}}{{\delta _t^2\beta _t^2}}}\\
\nonumber
&\;\;\;\;\;  + \frac{1}{{2}}\sum\limits_{t = 1}^T {\sum\limits_{i = 1}^n {( {\frac{{a{p^2}F_1^2{\alpha _t}}}{{\delta _t^2}} - \psi _t^2} )} }\mathbf{E}[ {{\| \psi _{i,t} \|}^2}] \\
\nonumber
&\;\;\;\;\;  + \sum\limits_{t = 1}^{T - 2} {\sum\limits_{i = 1}^n {2np\tau {F_1}\| {\varepsilon _{i,t}^e} \|\sum\limits_{k = 0}^{T - t - 2} {\frac{{{\lambda ^k}}}{{{\delta _{k + t + 2}}}}} } } \\
\nonumber
&\;\;\;\;\;  + \sum\limits_{t = 1}^{T - 2} {\sum\limits_{i = 1}^n {2np\tau {F_1^2}\| {\varepsilon _{i,t}^e} \|\sum\limits_{k = 0}^{T - t - 2} {\frac{{{\lambda ^k}}}{{{\delta _{k + t + 2}}{\beta _{k + t + 2}}}}} } } \\
\nonumber
&\;\;\;\;\;  + n{F_1}{F_2}\sum\limits_{t = 1}^T \frac{{\psi _t^1 + {\delta _t}}}{{{\beta _t}}} + n{F_2}\sum\limits_{t = 1}^T {( \psi _t^1 + 5{\delta _t} )} + \frac{{2nR{{( \mathbb{X} )}^2}}}{{{\alpha _{T + 1}}}}  \\
&\;\;\;\;\;  + \sum\limits_{t = 1}^T {\frac{{2nR{{( \mathbb{X} )}^2}}}{{{\alpha _{t + 1}}}}} {\mathbf{E}_{{\mathfrak{U}_t}}}( {{\xi _t} - {\xi _{t + 1}}} ) + \sum\limits_{t = 1}^T {\sum\limits_{i = 1}^n {\frac{{20 - a}}{{2a{\alpha _{t + 1}}}}} } {\| {\varepsilon _{i,t}^e} \|^2}. \label{theorem1-proof-eq29}
\end{flalign}
\end{small}%

Let $\Psi  = 1/{\gamma _1} + \sum\nolimits_{t = 1}^T {(2\delta _t+ {\beta _t} + 4a{F_2}{\varpi _3}{\alpha _t} + a{p^2}F_1^2{\alpha _t}/\delta _t^2)} $ and $\mu _{ij}^0 = \sum\nolimits_{t = 1}^T {{{[{c_{i,t}}({x_{j,t}})]}_ + }} /\Psi $.

We have\par\nobreak\vspace{-10pt}
\begin{small}
\begin{flalign}
{h_{ij}}( {\mu _{ij}^0} ) = \frac{{{{\| {\sum\nolimits_{t = 1}^T {{{[ {{c_{i,t}}( {{x_{j,t}}} )} ]}_ + }} } \|^2}}}}{{2\Psi}}. \label{theorem1-proof-eq30}
\end{flalign}
\end{small}%

From ${c_t}( x ) = {\rm{col}}\big( {{c_{1,t}}( x ),...,{c_{n,t}}( x )} \big)$, it holds that\par\nobreak\vspace{-10pt}
\begin{small}
\begin{flalign}
\mathbf{E}\Big[ {\sum\limits_{i = 1}^n {{\sum\limits_{j = 1}^n {\Big\| {\sum\limits_{t = 1}^T {{{[ {{c_{i,t}}( {{x_{j,t}}} )} ]}_ + }} } \Big\|^2}}} } \Big]
= \mathbf{E}\Big[ {\sum\limits_{j = 1}^n {{{\Big\| {\sum\limits_{t = 1}^T {{{[ {{c_t}( {{x_{j,t}}} )} ]}_ + }} } \Big\|^2}}} } \Big]. \label{theorem1-proof-eq31}
\end{flalign}
\end{small}%

From \eqref{ass-eq2a}, we have\par\nobreak\vspace{-10pt}
\begin{small}
\begin{flalign}
- \mathbf{E}[ {{\rm{Net}\mbox{-}\rm{Reg}}( {\{ {{x_{i,t}}} \},\{ y \}} )} ] \le 2{F_1}T. \label{theorem1-proof-eq32}
\end{flalign}
\end{small}%

From \eqref{ass-eq2b}, we have\par\nobreak\vspace{-10pt}
\begin{small}
\begin{flalign}
\| {\mu _{ij}^0} \| \le \frac{{{F_1}T}}{\Psi}. \label{theorem1-proof-eq33}
\end{flalign}
\end{small}%

We have\par\nobreak\vspace{-10pt}
\begin{small}
\begin{flalign}
\nonumber
&\;\;\;\;\; \mathbf{E}\Big[ {\sum\limits_{t = 1}^T {\| {{{[ {{c_t}( {{x_{j,t}}} )} ]}_ + }} \|} } \big] \le \mathbf{E}\big[ {\sum\limits_{t = 1}^T {{{\| {{{[ {{c_t}( {{x_{j,t}}} )} ]}_ + }} \|_1}}} } \Big]\\
& = \mathbf{E}\Big[ {{{\| {\sum\limits_{t = 1}^T {{{[ {{c_t}( {{x_{j,t}}} )} ]}_ + }} } \|_1}}} \big] \le \mathbf{E}\big[ {\sqrt m \| {\sum\limits_{t = 1}^T {{{[ {{c_t}( {{x_{j,t}}} )} ]}_ + }} } \|} \Big]. \label{theorem1-proof-eq34}
\end{flalign}
\end{small}%

From \eqref{theorem1-eq1} and \eqref{theorem1-proof-eq20}, we have\par\nobreak\vspace{-10pt}
\begin{small}
\begin{flalign}
\nonumber
& \sum\limits_{t = 1}^T {( {2\delta _t + {\beta _t} + 4a{F_2}{\varpi _3}{\alpha _t} + \frac{{a{p^2}F_1^2{\alpha _t}}}{{\delta _t^2}}} )} \le \frac{{2r( \mathbb{X} ){T^{1 - {g_3}}}}}{{1 - {g_3}}}\\
&  + \frac{{2{T^{1 - {g _2}}}}}{{1 - {g _2}}} + \frac{{a{F_2}{\varpi _3}r{{( \mathbb{X} )}^2}{T^{1 - {g_1}}}}}{{5{p^2}F_1^2( {1 - {g_1}} )}} + \frac{{a{T^{1 - {g _1} + 2{g _3}}}}}{20({1 - {g _1} + 2{g _3}})}. \label{theorem1-proof-eq35}
\end{flalign}
\end{small}%
Substituting ${\mu _i} = \mu _{ij}^0$ into \eqref{theorem1-proof-eq29}, combining \eqref{theorem1-proof-eq18}--\eqref{theorem1-proof-eq22.1}, \eqref{theorem1-proof-eq29}--\eqref{theorem1-proof-eq33} and \eqref{theorem1-proof-eq35}, and choosing $a = 20$ yields
\par\nobreak\vspace{-10pt}
\begin{small}
\begin{flalign}
\nonumber
&\mathbf{E}\Big[ {{{\Big( {\frac{1}{n}\sum\limits_{i = 1}^n {\| {\sum\limits_{t = 1}^T {{{[ {{c_t}( {{x_{i,t}}} )} ]}_ + }} } \|} } \Big)^2}}} \Big] \le \mathbf{E}\Big[ {\frac{1}{n}\sum\limits_{i = 1}^n {{{\| {\sum\limits_{t = 1}^T {{{[ {{c_t}( {{x_{i,t}}} )} ]}_ + }} } \|^2}}} } \Big]\\
\nonumber
& \le 4{F_1}{F_2}{\varpi _2}T + 2n\Big( 1 + \frac{{2r( \mathbb{X} ){T^{1 - {g_3}}}}}{{1 - {g_3}}} + \frac{{2{T^{1 - {g _2}}}}}{{1 - {g _2}}} \\
\nonumber
&  + \frac{{a{F_2}{\varpi _3}r{{( \mathbb{X} )}^2}{T^{1 - {g_1}}}}}{{5{p^2}F_1^2( {1 - {g_1}} )}} + \frac{{a{T^{1 - {g _1} + 2{g _3}}}}}{20({1 - {g _1} + 2{g _3}})} \Big) \Big( {2{F_1}T} + 2{F_2}{\varpi _2} \\
\nonumber
&  + \frac{{2F_1^2{T^{{g_2}}}}}{{{g_2}}}  + \frac{{4nF_2^2r( \mathbb{X} ){T^{1 - {g_3}}}}}{{1 - {g_3}}} 
+ \frac{{4p{F_1}{\varpi _1}{T^{{g_3}}}}}{{\lambda ( {1 - \lambda } )r( \mathbb{X} )}} + \frac{{2pF_1^2{\varpi _1}{T^{{g_2} + {g_3}}}}}{{\lambda ( {1 - \lambda } )r( \mathbb{X} )}}\\
\nonumber
&  + \frac{{2{F_2}{\varpi _3}r{{( \mathbb{X} )}^2}{T^{1 - {g_1}}}}}{{{p^2}F_1^2( {1 - {g_1}} )}} + \frac{{9{T^{1 - {g_1} + {g_3}}}}}{{4( {1 - {g_1} + {g_3}} )}} + \frac{{F_1^2{T^{1 - {g_1} + 2{g_2} + 2{g_3}}}}}{{2( {1 - {g_1} + 2{g_2} + 2{g_3}} )}}\\
\nonumber
&  + \frac{{2{n^2}{\tau ^2}{T^{1 - {g_1} + 2{g_3}}}}}{{{{( {1 - \lambda } )}^2}( {1 - {g_1}} )}} + \frac{{{n^2}{\tau ^2}F_1^2{T^{1 - {g_1} + 2{g_2} + 2{g_3}}}}}{{2{{( {1 - \lambda } )}^2}( {1 - {g_1}} )}} \\
\nonumber
&  + \frac{{{F_1}{F_2}\big( {R( \mathbb{X} ) + 2r( \mathbb{X} )} \big){T^{1 + {g_2} - {g_3}}}}}{{2( {1 + {g_2} - {g_3}} )}} + \frac{{{F_2}\big( {R( \mathbb{X} ) + 6r( \mathbb{X} )} \big){T^{1 - {g_3}}}}}{{1 - {g_3}}}\\
&  + \frac{{320{p^2}F_1^2R{{( \mathbb{X} )}^2}{T^{{g_1}}}}}{{r{{( \mathbb{X} )}^2}}} + \frac{{80{p^2}F_1^2R{{( \mathbb{X} )}^2}\log ( T )}}{{r{{( \mathbb{X} )}^2}}}\Big). \label{theorem1-proof-CCV}
\end{flalign}
\end{small}%
From \eqref{theorem1-proof-eq34} and \eqref{theorem1-proof-CCV}, we know that \eqref{theorem1-eq3} holds.

\subsection{The Proof of Theorem 2}
Denote ${\varpi _4} = {( {20{p^2}F_1^2/r{{( \mathbb{X} )}^2}\mu } )^{1/( {1 - {g _1}} )}}$ and ${\varpi _5} = \lceil {{\varpi _4} } \rceil -1 $, where $\lceil  \cdot  \rceil $ is the ceiling function.

From Lemma 3 in \cite{Yi2023}, we know that ${{\hat l}_{i,t}}$ is strongly convex the parameter $\mu $ since ${ l}_{i,t}$ is strongly convex the parameter $\mu $. 
Then, from \eqref{def-eq1}, \eqref{lemma5-proof-eq9} can be replaced by\par\nobreak\vspace{-10pt}
\begin{small}
\begin{flalign}
\nonumber
&\;\;\;\;\; {{\hat l}_{i,t}}( {{e_{i,t}}} ) - {{\hat l}_{i,t}}( {{{\hat y}_t}} )\\
& \le {\mathbf{E}_{{\mathfrak{U}_t}}}[ \frac{{p{F_1}}}{{{\delta _t}}}\| {\varepsilon _{i,t}} \| + \langle {\hat \partial {l_{i,t}}( {{e_{i,t}}} ),{e_{i,t + 1}} - {{\hat y}_t}} \rangle - \frac{\mu }{2}\| {{{\hat y}_t} - {e_{i,t}}} \|^2]. \label{theorem2-proof-eq1}
\end{flalign}
\end{small}%
Then, \eqref{theorem1-proof-eq4} can be replaced by\par\nobreak\vspace{-10pt}
\begin{small}
\begin{flalign}
\nonumber
&\frac{1}{n}\sum\limits_{i = 1}^n {\big( {\frac{1}{{{\alpha _{t + 1}}}}( {{{\| {{{\hat y}_t} - {z_{i,t + 1}}} \|}^2} - {{\| {{{\hat y}_{t + 1}} - {z_{i,t + 2}}} \|}^2}} ) - \mu {{\| {{{\hat y}_t} - {e_{i,t}}} \|}^2}} \big)}\\
\nonumber
& = \frac{1}{n}\sum\limits_{i = 1}^n \Big( {\frac{1}{{{\alpha _t}}}{{\| {{{\hat y}_t} - {z_{i,t + 1}}} \|}^2} - \frac{1}{{{\alpha _{t + 1}}}}{{\| {{{\hat y}_{t + 1}} - {z_{i,t + 2}}} \|}^2}}  \\
\nonumber
&\;\;\;\;\; + ( {\frac{1}{{{\alpha _{t + 1}}}} - \frac{1}{{{\alpha _t}}}} ){{\| {{{\hat y}_t} - \sum\limits_{j = 1}^n {{{[ {{W_t}} ]}_{ij}}{e_{j,t}}} } \|^2}} - \mu {{\| {{{\hat y}_t} - {e_{i,t}}} \|}^2} \Big)\\
\nonumber
& \le \frac{1}{n}\sum\limits_{i = 1}^n \big( {\frac{1}{{{\alpha _t}}}{{\| {{{\hat y}_t} - {z_{i,t + 1}}} \|}^2} - \frac{1}{{{\alpha _{t + 1}}}}{{\| {{{\hat y}_{t + 1}} - {z_{i,t + 2}}} \|}^2}}  \\
&\;\;\;\;\;  + ( {\frac{1}{{{\alpha _{t + 1}}}} - \frac{1}{{{\alpha _t}}} - \mu } ){{\| {{{\hat y}_t} - {e_{i,t}}} \|}^2} \big), \label{theorem2-proof-eq2}
\end{flalign}
\end{small}%
where the equality holds due to \eqref{Algorithm1-eq1}, and the inequality holds since $\sum\nolimits_{i = 1}^n {{{[ {{W_t}} ]}_{ij}}}  = \sum\nolimits_{j = 1}^n {{{[ {{W_t}} ]}_{ij}}}  = 1$ and ${\|  \cdot  \|^2}$ is convex.

From \eqref{ass-eq1}, we have
\par\nobreak\vspace{-10pt}
\begin{small}
\begin{flalign}
\nonumber
&\;\;\;\;\;\frac{1}{n}\sum\limits_{i = 1}^n {\sum\limits_{t = 1}^T {( {\frac{1}{{{\alpha _t}}}{{\| {{{\hat y}_t} - {z_{i,t + 1}}} \|}^2} - \frac{1}{{{\alpha _{t + 1}}}}{{\| {{{\hat y}_{t + 1}} - {z_{i,t + 2}}} \|}^2}} )} }\\
&  \le \frac{1}{n}\sum\limits_{i = 1}^n {\frac{1}{{{\alpha _1}}}} {\| {{{\hat y}_1} - {z_{i,2}}} \|^2} \le \frac{{4R{{( \mathbb{X} )}^2}}}{{{\alpha _1}}}. \label{theorem2-proof-eq2.1}
\end{flalign}
\end{small}%
When $t \ge {\varpi _5}$, it holds that\par\nobreak\vspace{-10pt}
\begin{small}
\begin{flalign}
\nonumber
&\;\;\;\;\; \frac{1}{{{\alpha _{t + 1}}}} - \frac{1}{{{\alpha _t}}} - \mu \\
\nonumber
& = \frac{{20{p^2}F_1^2( {t + 2} )}}{r{{( \mathbb{X} )}^2}{{{( {t + 2} )}^{1 - {g _1}}}}} - \frac{{20{p^2}F_1^2( {t + 1} )}}{r{{( \mathbb{X} )}^2}{{{( {t + 1} )}^{1 - {g _1}}}}} - \mu \\
& \le \frac{{20{p^2}F_1^2}}{{r{{( \mathbb{X} )}^2}{({t+1})^{1 - {g _1}}}}} - \mu  \le 0. \label{theorem2-proof-eq3}
\end{flalign}
\end{small}%
From \eqref{ass-eq1} and \eqref{theorem1-eq1}, we have\par\nobreak\vspace{-10pt}
\begin{small}
\begin{flalign}
\nonumber
&\;\;\;\;\;\sum\limits_{t = 1}^{{\varpi _5 - 1}} {( {\frac{1}{{{\alpha _{t + 1}}}} - \frac{1}{{{\alpha _t}}} - \mu } )} {\| {{\hat y}_t - {e_{i,t}}} \|^2} \\
&\le 4( {{\varpi _5} - 1} ){[ {\frac{{20{p^2}F_1^2}}{{r{{( \mathbb{X} )}^2}}} - \mu } ]_ + }R{( \mathbb{X} )^2}. \label{theorem2-proof-eq4}
\end{flalign}
\end{small}%

Similar to the way to get \eqref{theorem1-proof-regret} and \eqref{theorem1-proof-CCV}, from \eqref{theorem1-eq1}, and \eqref{theorem2-proof-eq2}--\eqref{theorem2-proof-eq4}, we have 
\par\nobreak\vspace{-10pt}
\begin{small}
\begin{flalign}
\nonumber
&{\mathbf{E}}[ {{\rm{Net}\mbox{-}\rm{Reg}}( {\{ {{x_{i,t}}} \},{y_{[T]}}} )} ]\\
\nonumber
& \le 2{F_2}{\varpi _2} + \frac{{2F_1^2{T^{{g _2}}}}}{{{g _2}}} + \frac{{4p{F_1}{\varpi _1}{T^{{g_3}}}}}{{\lambda ( {1 - \lambda } )r( \mathbb{X} )}} + \frac{{2pF_1^2{\varpi _1}{T^{{g_2} + {g_3}}}}}{{\lambda ( {1 - \lambda } )r( \mathbb{X} )}}\\
\nonumber
&  + \frac{{2{F_2}{\varpi _3}r{{( \mathbb{X} )}^2}{T^{1 - {g_1}}}}}{{{p^2}F_1^2( {1 - {g_1}} )}} + \frac{{9{T^{1 - {g_1} + {g_3}}}}}{{4( {1 - {g_1} + {g_3}} )}} + \frac{{F_1^2{T^{1 - {g_1} + 2{g_2} + 2{g_3}}}}}{{2( {1 - {g_1} + 2{g_2} + 2{g_3}} )}} \\
\nonumber
&  + \frac{{2{n^2}{\tau ^2}{T^{1 - {g_1} + 2{g_3}}}}}{{{{( {1 - \lambda } )}^2}( {1 - {g_1}} )}} + \frac{{{n^2}{\tau ^2}F_1^2{T^{1 - {g_1} + 2{g_2} + 2{g_3}}}}}{{2{{( {1 - \lambda } )}^2}( {1 - {g_1}} )}} \\
\nonumber
&  + \frac{{{F_1}{F_2}\big( {R( \mathbb{X} ) + 2r( \mathbb{X} )} \big){T^{1 + {g_2} - {g_3}}}}}{{2( {1 + {g_2} - {g_3}} )}} + \frac{{{F_2}\big( {R( \mathbb{X} ) + 6r( \mathbb{X} )} \big){T^{1 - {g_3}}}}}{{1 - {g_3}}}\\
\nonumber
&  + \frac{{160{p^2}F_1^2R{{( \mathbb{X} )}^2}}}{{r{{( \mathbb{X} )}^2}}} + 4( {{\varpi _5} - 1} ){[ {\frac{{20{p^2}F_1^2}}{{r{{( \mathbb{X} )}^2}}} - \mu } ]_ + }R{( \mathbb{X} )^2}\\
&  + \frac{{80{p^2}F_1^2R{{( \mathbb{X} )}^2}{T^{{g_1}}}{P_T}}}{{r{{( \mathbb{X} )}^2}}}  + \frac{{80{p^2}F_1^2R{{( \mathbb{X} )}^2}\log ( T )}}{{r{{( \mathbb{X} )}^2}}}. \label{theorem2-proof-regret}
\end{flalign}
\end{small}%
and
\par\nobreak\vspace{-10pt}
\begin{small}
\begin{flalign}
\nonumber
&\mathbf{E}\Big[ {{{\Big( {\frac{1}{n}\sum\limits_{i = 1}^n {\| {\sum\limits_{t = 1}^T {{{[ {{c_t}( {{x_{i,t}}} )} ]}_ + }} } \|} } \Big)^2}}} \Big] \\
\nonumber
& \le 4{F_1}{F_2}{\varpi _2}T + 2n\Big( 1 + \frac{{2r( \mathbb{X} ){T^{1 - {g_3}}}}}{{1 - {g_3}}} + \frac{{2{T^{1 - {g _2}}}}}{{1 - {g _2}}} \\
\nonumber
&  + \frac{{a{F_2}{\varpi _3}r{{( \mathbb{X} )}^2}{T^{1 - {g_1}}}}}{{5{p^2}F_1^2( {1 - {g_1}} )}} + \frac{{a{T^{1 - {g _1} + 2{g _3}}}}}{20({1 - {g _1} + 2{g _3}})} \Big) \Big( {2{F_1}T} + 2{F_2}{\varpi _2} \\
\nonumber
&  + \frac{{2F_1^2{T^{{g_2}}}}}{{{g_2}}}  + \frac{{4nF_2^2r( \mathbb{X} ){T^{1 - {g_3}}}}}{{1 - {g_3}}} 
+ \frac{{4p{F_1}{\varpi _1}{T^{{g_3}}}}}{{\lambda ( {1 - \lambda } )r( \mathbb{X} )}} + \frac{{2pF_1^2{\varpi _1}{T^{{g_2} + {g_3}}}}}{{\lambda ( {1 - \lambda } )r( \mathbb{X} )}}\\
\nonumber
&  + \frac{{2{F_2}{\varpi _3}r{{( \mathbb{X} )}^2}{T^{1 - {g_1}}}}}{{{p^2}F_1^2( {1 - {g_1}} )}} + \frac{{9{T^{1 - {g_1} + {g_3}}}}}{{4( {1 - {g_1} + {g_3}} )}} + \frac{{F_1^2{T^{1 - {g_1} + 2{g_2} + 2{g_3}}}}}{{2( {1 - {g_1} + 2{g_2} + 2{g_3}} )}}\\
\nonumber
&  + \frac{{2{n^2}{\tau ^2}{T^{1 - {g_1} + 2{g_3}}}}}{{{{( {1 - \lambda } )}^2}( {1 - {g_1}} )}} + \frac{{{n^2}{\tau ^2}F_1^2{T^{1 - {g_1} + 2{g_2} + 2{g_3}}}}}{{2{{( {1 - \lambda } )}^2}( {1 - {g_1}} )}} \\
\nonumber
&  + \frac{{{F_1}{F_2}\big( {R( \mathbb{X} ) + 2r( \mathbb{X} )} \big){T^{1 + {g_2} - {g_3}}}}}{{2( {1 + {g_2} - {g_3}} )}} + \frac{{{F_2}\big( {R( \mathbb{X} ) + 6r( \mathbb{X} )} \big){T^{1 - {g_3}}}}}{{1 - {g_3}}}\\
\nonumber
&  + \frac{{160{p^2}F_1^2R{{( \mathbb{X} )}^2}}}{{r{{( \mathbb{X} )}^2}}} + 4( {{\varpi _5} - 1} ){[ {\frac{{20{p^2}F_1^2}}{{r{{( \mathbb{X} )}^2}}} - \mu } ]_ + }R{( \mathbb{X} )^2} \\
&  + \frac{{80{p^2}F_1^2R{{( \mathbb{X} )}^2}\log ( T )}}{{r{{( \mathbb{X} )}^2}}}\Big). \label{theorem2-proof-CCV}
\end{flalign}
\end{small}%
From \eqref{theorem2-proof-regret} and \eqref{theorem2-proof-CCV}, we know that \eqref{theorem2-eq1} and \eqref{theorem2-eq2} hold.
\end{proof}

\bibliographystyle{IEEEtran}
\bibliography{reference_online}

\end{document}